\documentclass[graybox, envcountchap]{svmult}

\usepackage{mathptmx}      % selects Times Roman as basic font
\usepackage{amsmath}
\usepackage{amssymb}
\usepackage{color}
\usepackage{helvet}           % selects Helvetica as sans-serif font
\usepackage{courier}         % selects Courier as typewriter font
\usepackage{dirtree}
\usepackage{type1cm}    % activate if the above 3 fonts are 
\usepackage{makeidx}        % allows index generation
\usepackage{graphicx}        % standard LaTeX graphics tool when including figure files
\usepackage{subfig}

\usepackage{multicol}        % used for the two-column index
\usepackage[bottom]{footmisc}% places footnotes at page bottom

\usepackage{hyperref}
\hypersetup{colorlinks=true,urlcolor=blue}

\usepackage[misc]{ifsym}

\usepackage{cite}

\makeindex       % used for the subject index
\def\qq{\qquad}
\def\cm{\hspace*{1cm}}
\def\inch{\hspace*{1in}}

\def\noi{\noindent}
\usepackage{color}

\def\Jl#1#2{{\it #1}\ {\bf #2},\ }

\def\ApJ#1 {\Jl{Astroph. J.}{#1}}
\def\CQG#1 {\Jl{Class. Quantum Grav.}{#1}}
\def\DAN#1 {\Jl{Dokl. AN SSSR}{#1}}
\def\GC#1 {\Jl{Grav. Cosmol.}{#1}}
\def\GRG#1 {\Jl{Gen. Rel. Grav.}{#1}}
\def\JETF#1 {\Jl{Zh. Eksp. Teor. Fiz.}{#1}}
\def\JETP#1 {\Jl{Sov. Phys. JETP}{#1}}
\def\JHEP#1 {\Jl{JHEP}{#1}}
\def\JMP#1 {\Jl{J. Math. Phys.}{#1}}
\def\NPB#1 {\Jl{Nucl. Phys. B}{#1}}
\def\NP#1 {\Jl{Nucl. Phys.}{#1}}
\def\PLA#1 {\Jl{Phys. Lett. A}{#1}}
\def\PLB#1 {\Jl{Phys. Lett. B}{#1}}
\def\PRD#1 {\Jl{Phys. Rev. D}{#1}}
\def\PRL#1 {\Jl{Phys. Rev. Lett.}{#1}}

\def\al{&\!}
\def\lal{&&\!\! {}}
\def\eq{Eq.\,}
\def\eqs{Eqs.\,}
\def\beq{\begin{equation}}
\def\eeq{\end{equation}}
\def\bear{\begin{eqnarray}}
\def\bearr{\begin{eqnarray} \lal}
\def\ear{\end{eqnarray}}
\def\earn{\nonumber \end{eqnarray}}
\def\nn{\nonumber\\ {}}
\def\nnv{\nonumber\\[5pt] {}}
\def\nnn{\nonumber\\ \lal }
\def\nnnv{\nonumber\\[5pt] \lal }
\def\yy{\\[5pt] {}}
\def\yyy{\\[5pt] \lal }
\def\eql{\al =\al}

\def\dst{\displaystyle}
\def\tst{\textstyle}
\def\fracd#1#2{{\dst\frac{#1}{#2}}}
\def\fract#1#2{{\tst\frac{#1}{#2}}}
\def\Half{{\fracd{1}{2}}}
\def\half{{\fract{1}{2}}}
\def\e{{\,\rm e}}
\def\d{\partial}

\def\const{{\rm const}}
\def\eps{\varepsilon}
\def\ep{\epsilon}

\def\rf{\eqref}
\def\eqn{\eq\rf} 
\def\mn{_{\mu\nu}}
\def\MN{^{\mu\nu}}
\def\mN{_\mu^\nu}
\def\nM{_\nu^\mu}

\def\R{{\mathbb R}}

\def\cK{{\cal K}}
\def\cO{{\cal O}}
\def\sF{{}^*\!F}
\def\kappa{\varkappa}
\def\GR{general relativity}
\def\sph{spherically symmetric}
\def\ssph{static, spherically symmetric}

\def\bh{black hole}
\def\bhs{black holes}
\def\wh{wormhole}

\def\asflat{asymptotically flat} 
\def\emag{electromagnetic} 
\def\Scw{Schwarz\-schild}
\def\RN{Reiss\-ner-Nord\-str\"om}

\begin{document}

\title{Regular black holes sourced by nonlinear electrodynamics}
\author{Kirill A. Bronnikov}{}
\institute{Kirill A. Bronnikov \at 
		 {VNIIMS, Ozyornaya ul. 46, Moscow 119361, Russia;\\
            Institute of Gravitation and Cosmology, Peoples' Friendship University of Russia\\ 
            (RUDN University),  ul. Miklukho-Maklaya 6, Moscow 117198, Russia;\\
	   	 National Research Nuclear University ``MEPhI'', 
	   	 Kashirskoe sh. 31, Moscow 115409, Russia.}\\
\email{kb20@yandex.ru}
}
\maketitle

\abstract{This chapter is a brief review on the existence and basic properties of \ssph\ regular \bh\ 
  solutions of \GR\ where the source of gravity is represented by nonlinear \emag\ fields with the 
  Lagrangian function $L$ depending on the single invariant $f = F\mn F\MN$ or on two variables:
  either $L(f, h)$, where $h = \sF\mn F\MN$, where $\sF\mn$ is the Hodge dual of $F\mn$,
  or $L(f, J)$, where $J = F\mn F^{\nu\rho} F_{\rho\sigma} F^{\sigma\mu}$. 
  A number of no-go theorems are discussed, revealing the conditions under which the space-time 
  cannot have a regular center, among which the theorems concerning $L(f,J)$ theories are probably
  new. These results concern both regular \bhs\ and regular particlelike or starlike objects (solitons) 
  without horizons. Thus, a regular center in solutions with an electric charge $q_e\ne 0$ 
  is only possible with NED having no Maxwell weak field limit. Regular solutions with $L(f)$ 
  and $L(f, J)$ nonlinear electrodynamics (NED), possessing a correct (Maxwell) weak-field limit,
  are possible if the system contains only a magnetic charge $q_m \ne 0$. It is shown, however, 
  that in such solutions the causality and unitarity as well as dynamic stability conditions are 
  inevitably violated in a neighborhood of the center. Some particular examples are discussed. 
  }
  
% ===============================
\section{Introduction}\label{sec1}
% ===============================

  Nonlinear electrodynamics (NED) as a generalization of Maxwell's theory was proposed
  in the 1930s: M. Born and L. Infeld's formulated a theory able to remove the central singularity 
  of the \emag field of a point charge as well as its energy divergence \cite{B-Inf}.
  Another version of NED was put forward by W. Heisenberg and H. Euler while taking into
  consideration high-energy quantum processes with photons, such as pair creation \cite{EuH}.
  Much later, J. Plebanski \cite{Pleb} developed a more general formulation of NED in 
  special relativity, admitting an arbitrary function of the electromagnetic invariants. 

  More recently, the interest in NED received a new support when it was discovered that 
  a Born-Infeld-like theory appears in the weak-field limit of some models of string theory 
  \cite{tse-1, tse-2, seiberg}. It has also turned out that NED can be a material source of gravity
  able to lead to nonsingular geometries of interest, such as regular black holes (BHs) and 
  solitonlike configurations without horizons in the framework of general relativity (GR) and various 
  alternative theories. Let us also mention one more recent application of NED, namely, using it 
  as one of the sources of gravity in Simpson-Visser-like (black-bounce) space-times 
  \cite{simp-18, fran-21, lobo-20} that are regular models simulating some expected effects 
  of quantum gravity on the classical level \cite{k-22a, k-22b, canate-22}.   
 
  This paper is devoted to NED application for obtaining regular BHs and solitons (monopoles)
  in GR. We will reproduce a number of well-known results in a somewhat pedagogical manner
  and also present some new observations. We will restrict ourselves to the simplest models 
  assuming spherical symmetry, and also mostly focus on NED theories with Lagrangians 
  of the form $L = L(f)$, where  $f = F\mn F\MN$, and $F\mn$ is the \emag\ field tensor.
  Then we will briefly discuss similar problems in some extensions of $L(f)$ theories:
  those with $L(f, h)$, where $h = \sF\mn F\MN$, where 
  $\sF\mn = \half \sqrt{-g} \eps_{\mu\nu\rho\sigma} F^{\rho\sigma}$
  is the Hodge dual of $F\mn$, and those with $L(f, J)$, where $J$ is quartic with respect 
  to $F\mn$ \cite{denis-19, sokolov-21}: $J = F\mn F^{\nu\rho} F_{\rho\sigma} F^{\sigma\mu}$.  
  
  When considering the NED-GR system in spherical symmetry, there are only two possible kinds 
  of electromagnetic fields: radial electric fields and radial (monopole) magnetic ones. Two important 
  circumstances should be taken into account. The first one is (in general) the absence of 
  duality between electric and magnetic fields, so that solutions to the field equations containing 
  these fields in the framework of the same NED theory will be quite different. Instead, there 
  emerges the so-called FP duality that connects electric and magnetic solutions for {\it different\/}
  NED theories but involving the same space-time metric. The second circumstance is that 
  it is insufficient to require finite values of the electric field itself and the electric field energy of 
  a point charge in order to obtain a regular space-time: its regularity imposes more stringent 
  requirements on NED, which cannot be satisfied, for example, by the Born-Infeld theory. 
     
  NED-GR solutions with electric or magnetic fields are currently widely discussed, probably 
  beginning with finding a general form of an electric solution by Pellicer and Torrence \cite{Pel-T}. 
  Later on, a no-go theorem was proved \cite{B-Shi, BShi2}, showing that if NED is specified by a 
  Lagrangian function $L(f)$ having a Maxwell weak-field limit ($L\sim f$ as $f\to 0$), a \ssph\ 
  solution of GR with an electric field cannot have a regular center. This theorem was extended 
  to include static dyonic configurations, involving both electric and magnetic fields \cite{k-NED}, 
  and it was also proved \cite{k-NED, B-comment} that in any electric solutions describing 
  systems with or without horizons (i.e., BH or solitonic ones), containing a regular center 
  and a flat infinity with a Reissner-Nordstr\"om (RN) asymptotic behavior, different NED 
  theories are valid at large and small $r$. The present paper describes this issue in detail. 

  It was also shown \cite{k-NED} that purely magnetic regular configurations, both BH and 
  solitonic ones, can exist and are easily obtained if $L(f)$ tends to a finite limit as $f  \to \infty$. 
  Electric solutions with the same metric can also be found, but they suffer multivaluedness of 
  $L(f)$ and inevitably exhibit infinite blueshifts of traveling photons on some surfaces
  \cite{k-NED}.

  Many further results of interest are known. In particular, the properties and examples 
  of \ssph\ dyonic NED-GR space-times were studied 
  \cite{17-dyon, krugl-19a, krugl-19b, mkr-22, kru-20, yang-22};
  a kind of phase transition was discussed, allowing one to circumvent the above no-go 
  theorem on electric solutions \cite{Bur1}; the \ssph\ solutions were extended to include 
  a nonzero cosmological constant $\Lambda$ \cite{Mat-09}; the thermodynamic properties 
  of regular NED BHs were investigated (see \cite{Bret-05, Kru-16, FW-16, kru-21, balart-21} 
  and references therein); cylindrically
  \cite{we-02} and axially \cite{Bam-13, Dym-15a, tosh-17, gdiaz-22, kubi-22} symmetric 
  (rotating) NED-GR configurations were found and studied, as well as evolving wormhole models 
  \cite{Arel-06, Boe-07, Arel-09, k-18}. 
  (Note that static \wh\ models with NED as a source are impossible because this kind of matter
  respects the weak energy condition.) Furthermore, the stability properties of NED BHs were
  investigated in \cite{Mor-03, Bret-05s, Jin-14}, and quantum effects in their fields in 
  \cite{Mat-02, Mat-13}. One should also mention a number of studies of special cases of both
  electric and magnetic solutions, their potential observational properties like gravitational lensing, 
  particle motion and matter accretion in the fields of NED BHs as compared to their counterparts 
  in scalar-tensor, $f(R)$ and multidimensional theories of gravity, consideration of NED with 
  dilaton-like interactions, non-Abelian fields, different constructions with thin shells, etc., 
  but the corresponding list of references would be too long. For a recent brief review on NED 
  with and without relation to gravitational theories see \cite{sorokin-21}.

  The most relevant to the present subject, regular BHs, are the recent results obtained by 
  Bokuli\'c, Smoli\'c and Juri\'c \cite{bokulic-21, bokulic-22} who have proved a number of
  no-go theorems in NED-GR solutions with NED Lagrangians of the form $L(f, h)$.  Let us 
  mention that this wide class of theories contains, among others, the Born-Infeld and 
  Heisenberg-Euler theories. With all these no-go theorems, it seems that regular magnetic 
  BHs with $L = L(f)$ are the only kind of regular BHs that can be found among NED-GR 
  solutions with an asymptotically Maxwell NED, although some opportunities are still 
  remaining unexplored.

  In this paper, we will discuss in detail the existence and main properties of \ssph\ regular \bhs\ 
  and solitons with $L(f)$ NED theory, and more briefly consider the same with Lagrangians 
  depending on two invariants, either $L(f,h)$ or  $L(f,J)$. We begin with discussing the 
  particular form of regularity and asymptotic conditions to be used (Section \ref{sec2}). Then, 
  in Section \ref{sec3}, we discuss the $L(f)$ NED-GR field equations in \ssph\ space-times. 
  Section \ref{sec4} is devoted to the regularity properties of \bh\ and soliton solutions to these 
  equations, their compatibility with the known NED unitarity and causality \cite{usov-11}
  and stability \cite{Mor-03} conditions as well as photon propagation in these space-times. 
  Some particular examples known in the literature are also discussed.
  Section \ref{sec5} presents some no-go theorems with $L(f,h)$ due to 
  \cite{bokulic-21, bokulic-22} and with $L(f, J)$, and the latter results seem to be new. 
  Section \ref{sec6} is a brief conclusion.    
        
  We use the following conventions: the units with $c= 8\pi G =1$; the metric signature $(+ - -\, - )$; 
  the curvature tensor $R^\sigma_{\ \mu\rho\nu} = \d_\nu \Gamma^\sigma_{\mu\nu} = \ldots$;
  the Ricci tensor $R\mn = R^\sigma_{\ \mu\sigma\nu}$, so that the Ricci scalar $R = g\MN R\mn >0$
  for de Sitter space-time. The Einstein equations are written in the form
\bearr                                                                   \label{EE}
  	  G\mN \equiv R\mN - \half \delta\mN R = - T\mN,  
\ear
  where $T\mN$ is the stress-energy tensor (SET) of matter, such that $T^t_t$ is the energy density.
    
% =======================================================
\section{Static \sph\ space-times. Regularity and asymptotic conditions}\label{sec2}
% =======================================================

  Before dealing with NED-Einstein equations, it makes sense to recall the conditions to be 
  fulfilled by the desirable solutions to these equations.

  Spherical symmetry is the simplest and natural assumption for descriptions of isolated bodies 
  when their precise shape and possible rotation are regarded insignificant. The physical fields of 
  any island-like objects are approximately \sph\ far from these objects. 
  
  In the general case, one can write a \sph\ metric in the form (see, e.g., \cite{LL})
\beq            \label{ds0} 
	   ds^2 = \e^{2\gamma} dt^2 - \e^{2\alpha} dx^2 - r^2 d\Omega^2, \cm
						d\Omega^2 = d\theta^2+\sin^2 \theta d\phi^2.
\eeq
  In general, $\alpha, \gamma, r$ are functions of the radial coordinate $x$ and the time
  coordinate $t$. The quantity $r$ has the geometric meaning of the radius of a coordinate sphere 
  $x=\const,\ t=\const$, the so-called spherical radius, or it is sometimes called the areal radius 
  since the area of a coordinate sphere is equal to $4\pi r^2$. Let us note that in curved 
  space-time this radius $r$ has nothing to do with a distance from the center (as happens in 
  flat space-time), and there are many \sph\ space-times that contain no center at all, for example, 
  wormholes.
 
  In what follows we restrict ourselves to static space-times, such that $\alpha, \gamma, r$ 
  depend on $x$ only. There still remains the freedom of choosing the radial coordinate $x$ and 
  the possibility of its reparametrizations by replacing $x = x(x_{\rm new})$. The choice of the 
  radial coordinate can be fixed by postulating a relation between the functions $\alpha, \gamma, r$ 
  or by choosing some of them (or a function of some of them) as the coordinate. For example, very 
  often the radius $r$ is used as a coordinate, it is then called the Schwarzschild (or curvature) 
  radial coordinate.
  
  The convenient ``exponential'' notations in the metric \rf{ds0}, which simplify the appearance of many
  relations without fixing the radial coordinate, assume positive values of the corresponding quantities. 
  However, the coefficients $g_{tt}$ and $g_{xx}$ can change their sign, in particular, this happens at
  \bh\ horizons.  In such cases, it is helpful to use the so-called quasiglobal coordinate condition
  $\alpha+\gamma =0$, and with the notation $\e{2\gamma} = \e^{-2\alpha} = A(x)$, the metric 
  is written as 
\beq             \label{ds1}
                     ds^2 = A(x) dt^2 - \frac {dx^2}{A(x)} -r^2 (x) d\Omega^2.
\eeq  
  
\medskip\noi   
  {\bf Regularity.} A Riemannian space-time is generally called regular at a particular point $X$ if
  the Riemann tensor is well defined at $X$ (hence the metric functions must be at least twice 
  differentiable at $X$), and all algebraic curvature invariants are finite. (Other definitions 
  of regularity, involving differential invariants of the Riemann tensor, are sometimes used, 
  but the above definition is sufficient for our purposes.) 
  Hence, the metric \rf{ds0} (or \rf{ds1}) is manifestly regular at any point where $r\ne 0$ as long as 
  the functions $\alpha(x)$, $\gamma(x)$ and $r(x)$ (or $A(x)$ and $r(x)$) are sufficiently smooth. 
  A point where $r=0$ requires special attention because the metric becomes degenerate there, 
  hence it is a singular point of the spherical coordinate system used in \rf{ds0} or \rf{ds1}. 
  Furthermore, a space-time as a whole (and in particular, a \bh\ space-time) is called regular
  if all its points are regular.     

  Very often, to verify regularity of a particular metric of the form \rf{ds0} or \rf{ds1}, one 
  directly calculates its basic invariants: the scalar curvature $R$, the Ricci tensor squared 
  $R\mn R\MN$, and the Kretschmann scalar (the Riemann tensor squared) 
  $\cK\,{=}\,R_{\alpha\beta\gamma\delta} R^{\alpha\beta\gamma\delta}$. 
  However, for a static metric \rf{ds0}, it is quite sufficient and much easier to verify finiteness
  of the four independent components $R_{\alpha\beta}{}^{\gamma\delta}$ of the Riemann 
  tensor with two upper and two lower indices: 
\bear \label{K_i}
      K_1 = -R_{01}{}^{01} \eql \e^{-\alpha-\gamma} (\gamma'\e^{\gamma-\alpha})'
             =  \Half A'',     
\nn
      K_2 = -R_{02}{}^{02} = -R_{03}{}^{03}\eql \e^{-2\alpha}\frac {\gamma' r'}{r} 
      	  =  \frac{A' r'}{2r},
\nn
      K_3 = -R_{12}{}^{12} = -R_{13}{}^{13}\eql \frac{\e^{-\alpha}}{r} (\e^{-\alpha} r')'
             = \frac{1}{2r} (2~A r'' - A' r'),      
\nn
      K_4 = -R_{23}{}^{23}\eql  \frac 1 {r^2}(1 - \e^{-2\alpha} r'^2) 
      	  = \frac 1 {r^2}(1 - A r'^2)    .
\ear
  where the prime stands for $d/dx$ ($K_i$ in terms of the metric \rf{ds1} are given in each line 
  after the last equality sign). The point is that for \ssph\ metrics, as well as and in many other 
  important cases, the tensor $R_{\alpha\beta}{}^{\gamma\delta}$ is pairwise diagonal. Therefore, 
  all algebraic curvature invariants are linear, quadratic, cubic, etc., combinations of $K_i$ from 
  \rf{K_i} and are manifestly finite if $K_i$ are finite. Moreover, the Kretschmann scalar 
  is a sum of squares: 
\beq       \label{Kre}
        \cK = 4K_1^2 + 8K_2^2 + 8K_3^2 + 4K_4^2,
\eeq
  hence it is finite if and only if {\it each\/} $K_i$ is finite. Thus finiteness of all $K_i$ 
  is both necessary and sufficient condition of space-time regularity \cite{BR-book}. 
  
  It is important to note that all $K_i$ in \rf{K_i} are invariant (behave as scalars) under 
  reparametrizations of the $x$ coordinate, and the same is true for mixed components of 
  second-rank tensors, including the Ricci tensor $R\mN$ and the Einstein tensor 
  $G\mN= R\mN - \half \delta\mN R$. Thus the space-time regularity can be verified
  using $K_i$ in terms of any radial coordinate $x$.
  
  As follows from \rf{K_i}, regularity at $r=0$ requires not only finite values and 
  smoothness of $\alpha$ and $\gamma$, but also, due to the expression for $K_4$, 
\beq               \label{reg-c}
                     \e^{-2\alpha} r'^2 - 1 = \cO(r^2)\qq {\rm as} \ \ r\to 0.
\eeq  
  It is actually the local flatness condition, requiring a circumference to radius ratio 
  of $2\pi$ for small circles around the center. 
  
  One more important observation follows from \eq \rf{reg-c} for \bh\ space-times, in which
  $A(x)$ can become negative. The condition \rf{reg-c}, rewritten as $A r'^2 -1 = \cO(r^2)$, 
  cannot be satisfied if $A(x) < 0$, which happens in nonstatic regions of \sph\ \bhs\ 
  (also called T-regions) beyond their horizons. We see that {\sl the metric cannot be regular 
  in the limit $r\to 0$ in T-regions of \sph\ \bhs.} A regular center can only occur in a static 
  region where $A>0$.
  
  It also follows from \rf{K_i} that the metric \rf{ds1} is regular at apparent horizons that 
  correspond to regular zeros of the function $A(x)$ under the condition $r(x) >0$.

\medskip\noi  
  {\bf Asymptotics.} For an island-like system, it is natural to assume that the space-time 
  is \asflat, and far from the source of gravity there is an approximately \Scw\ gravitational 
  field characterized by a certain mass $m$. In terms of an arbitrary radial coordinate $x$ 
  it means that in the metric \rf{ds0}, under the appropriate choice of the time scale, 
\beq        \label{asflat1}      
               \e^{2\gamma(x)} = 1 - \frac{2m}{r(x)} + o(1/r)\qq {\rm as} \ \ r \to \infty
\eeq      
  In addition, one should require a correct circumference to radius ratio for large circles
  around the source of gravity, which leads to a condition similar to \rf{reg-c},
\beq        \label{asflat2} 
			 \e^{-2\alpha} r'^2 \to 1  \qq {\rm as} \ \ r \to \infty.
\eeq    
  A limit other than unity in \rf{asflat2} leads to a deficit or excess of the solid angle at infinity, 
  characterizing a global monopole space-time \cite{vil-sh}.
  
  In the presence of a nonzero cosmological constant $\Lambda$, the gravitational field far
  from its island-like source as asymptotically de Sitter (if $\Lambda >0$) or anti-de Sitter
  (if $\Lambda <0$), well described by the metric \rf{ds1} with $r=x$ and 
  $A = 1 - \Lambda r^2/3$.
  
% ===================================
\section{$L(f)$ NED coupled to general relativity. FP duality}\label{sec3}
% ===================================
\subsection{Field equations}
% ------------------------------

  Let us now consider self-gravitating \emag\ fields with the Lagrangian $L(f)$ in the framework
  of GR, so that the total action has the form
\beq            \label{S}
	S = \Half \int \sqrt {-g} d^4 x [R - L (f)], 	
\eeq
  where $R$ is the Ricci scalar, the invariant $f$ has the standard form 
  $f = F\mn F\MN = 2({\bf B}^2 - {\bf E}^2)$, where the 3-vectors $\bf E$ and $\bf B$ are the 
  electric field strength and magnetic induction, and $L(f)$ is an arbitrary function. 
  The \emag\ tensor $F\mn$ obeys the Maxwell-like equations, obtained from \rf{S} by variation 
  with respect to the 4-vector potential $A_\mu$,  and the Bianchi identities  
  for the dual field $\sF\MN$, following from the definition $F\mn = \d_\mu A_\nu - \d_\nu A_\mu$:
\beq                   \label{Max}
			\nabla_\mu (L_f F\MN) = 0, \qq      \nabla_\mu \sF\MN = 0.
\eeq   
  The corresponding SET is given by ($L_f \equiv dL/df$)
\beq             \label{SET1}
             T\mN = -2 L_f F_{\mu\alpha} F^{\nu\alpha} + \half \delta\mN L(f).
\eeq

  Let us assume spherical symmetry, with a metric of the general form \rf{ds0}. The only nonzero 
  components of $F\mn$ compatible with this symmetry are $F_{tr} =- F_{rt}$, representing a 
  radial electric field, and $F_{\theta\phi} = - F_{\phi\theta}$, corresponding to a radial magnetic field.
  From \rf{Max} it follows
\beq              \label{F_mn}
        r^2 \e^{\alpha + \gamma} L_f F^{tr} = q_e, \cm  F_{\theta\phi} = q_m\sin\theta,
\eeq
  where $q_e = \const$\ has the meaning of an electric charge, and $q_m = \const$\ is a magnetic charge. 
  Accordingly, the only nonzero SET components have the form
\bearr        \label{SET2}
        T^t_t = T^r_r = \half L + f_e L_f, \qq
%\nnn  
        T^\theta_\theta = T^\phi_\phi = \half  L - f_m L_f,
\ear
  where  
\bearr                   \label{ff}
	f_e = 2 E^2 = 2F_{tr}F^{rt} = \frac{2q_e^2}{L_f^2 r^4}, \qq
	f_m =2 B^2 = 2F_{\theta\phi}F^{\theta\phi} = \frac {2q_m^2}{r^4},
\ear
  so that $f = f_m - f_e$. Here, $E = |{\bf E}|$ and $B = |{\bf B}|$ are the absolute values of the 
  electric field strength and magnetic induction, measured by an observer at rest in our static space-time.
  
  The SET \rf{SET2} has two important properties $T_t^x =0$ and $T^t_t = T^x_x$. 
  The first one means the absence of radial energy flows, related to the absence of monopole
  electromagnetic radiation. The second one, due to the Einstein equations \rf{EE}, leads to
  $G^t_t = G^x_x$, and this equation is easily integrated if we use the \Scw\ radial coordinate, 
  $x\equiv r$ (see, e.g., \cite{LL}), leading to the relation $\alpha(r) + \gamma(r) = \const$. 
  With a proper choice of the time scale, we have $\alpha + \gamma =0$, and the metric can 
  be rewritten as 
\beq            \label{ds2} 
		ds^2 = A(r) dt^2 - \frac{dr^2}{A(r)} - r^2 d\Omega^2.
\eeq
  The other Einstein equation, $G^t_t = -T^t_t$, then reads
\beq
		A + A'r = 1 - \rho r^2
\eeq
  and can be rewritten in the integral form as
\beq          \label{A}
	A(r) = 1 -\frac{2M(r)}{r}, \qq       M(r) = \frac 12 \int \rho(r) r^2 dr,  
\eeq
  where $\rho(r) \equiv T^t_t$ is the energy density, and $M(r)$ is called the mass function,
  such that $M(\infty)$ is the \Scw\ mass in an \asflat\ space-time.
  It is a solution for $A(r)$ if $\rho(r)$ is known. Note, however, that a complete solution for 
  the system under consideration requires a knowledge of $L(f)$ and both electric 
  and magnetic fields as functions of $r$.

% ------------------------------------
\subsection{FP duality}
% ------------------------------------  
  
  NED with a Largangian function $L(f)$ is known to admit an alternative representation 
  obtained from the original one by a Legendre transformation \cite{Pel-T, sala-87, vag-14}: 
  to this end, the new tensor $P\mn = L_f F\mn$ is defined, with its invariant  
  $p = P\mn P\MN$. Then one considers the Hamiltonian-like quantity 
\beq               \label{H}
	H(p) = 2f L_f  - L = - 2T^t_t 
\eeq
  as a function of $p$. It is possible to use the function $H(p)$ to specify the whole theory. 
  The following relations are valid:
\beq                                 \label{fp}
          L = 2p H_p - H, \qq  L_f H_p = 1, \qq  f = p H_p^2, \qq   p = f L_f^2,
\eeq
  where $H_p \equiv dH/dp$. In terms of $H$ and $P\mn$, the SET reads 
\beq                 \label{SET3}
	 T\mN =  -2H_p P_{\mu\alpha} P^{\nu\alpha} + \delta\mN (p H_p -\half H).
\eeq

  In a \sph\ space-time with the metric \rf{ds2}, \eqs \rf{F_mn} are rewritten in the 
  P framework as
\beq                     \label{P_mn}
                     r^2 P^{tr} = q_e, \cm    H_p P_{\theta\phi} = q_m\sin\theta.                       
\eeq
  Let us also introduce the quantities $p_e$ and $p_m$ quite similar to $f_e$ and $f_m$:  
\beq                   \label{pp}
	p_e = 2 P_{tr}P^{rt} = \frac{2q_e^2}{r^4} \geq 0, \qq  
	p_m = 2 P_{\theta\phi}P^{\theta\phi} = \frac {2q_m^2}{H_p^2 r^4} \geq 0,
\eeq
  so that $p = p_m - p_e$, and then the SET \rf{SET3} is transformed to
\beq        \label{SET4}
        T^t_t = T^r_r = - \half H + p_m H_p, 
\qq
        T^\theta_\theta = T^\phi_\phi = -\half H - p_e H_p.
\eeq
  One can notice that the $F$ and $P$ formulations of the same theory are not always 
  equivalent \cite{k-NED, B-comment}. More precisely, a theory initially specified by $L(f)$
  is equivalently reformulated in the $P$ framework only in a range of $f$ where $f(p)$ is a 
  monotonic function. In any case, the main and physically preferred  formulation is the
  Lagrangian one since it directly follows from the least action principle. Later on we will
  confirm this statement in the discussion of photon motion in magnetic and electric
  solutions with the same metric.

  Now, comparing \rf{SET2} and \rf{SET4}, one can see that they coincide up 
  to the substitution
\beq          \label{dual}
		\{F\mn,\ f,\ L(f)\} \ \  \longleftrightarrow \ \ \{^*P\mn,\ -p,\ -H(p)\},
\eeq
  where $^*P\mn$ is the Hodge dual of $P\mn$, such that $^*P_{\theta\phi}= P_{tx}$.
  As long as the SETs coincide, all possible metrics satisfying the Einstein equations
  \rf{EE} should also coincide. This coincidence was described in \cite{k-NED} for static 
  systems and was named FP duality. In \cite{Mor-03} this kind of duality was extended 
  to general space-times and used for studying the stability of static solutiojns, 
  and in \cite{k-18} it was used while obtaining nonstatic \sph\ solutions to the 
  Einstein-NED equations. 
  
  It should be stressed that the FP duality connects solutions with the same metric but
  belonging to {\it different\/} NED theories. Only in the Maxwell theory, in which 
  $L = f = H = p$, the FP duality is the same as the conventional electric-magnetic duality. 
  
% =====================================
\section {Regular \bhs\ with $L = L(f)$} \label{sec4}
% =====================================
\subsection{Magnetic, electric and dyonic solutions} 
% ----------------------------------------------------------------

  {\bf Magnetic solutions} ($q_e = 0, q_m \ne 0$) can be found most easily. If the Lagrangian 
  $L(f)$ is specified, then, since now $f = 2q_m^2/r^4$, the density $\rho(r) = L/2$ is known 
  according to \rf{SET2}, and the metric function $A(r)$ is found by integration in \rf{A}. 
  
  If, on the contrary, we know $A(r)$ (or choose it by hand), then $\rho = L(f)/2$ is found
  from \rf{A}, leading to 
\beq  
                L(f(r)) = \frac{2}{r^2}[1 - (rA)'],  
\eeq        
   and $L(f)$ is restored since $f = 2q_m^2/r^4$. 
  
\medskip\noi  
  {\bf Electric solutions} ($q_e\ne 0, q_m=0$) can be obtained in quite a similar manner 
  if we use the Hamiltonian-like form of NED, see \eqs \rf{fp}--\rf{SET4}. In this case, 
  $p= -2q_e^2/r^4$, and if we specify $H(p) = -2\rho$, the mass function $M(r)$ is directly
  found, while $A(r)$ is obtained by integration in \rf{A}. If $A(r)$ is specified, 
  then \eq \rf{A} allows for finding $\rho(r) = -H(p)/2$.  
  
  However, if one starts with the Lagrangian $L(f)$ and seeks electric solutions, a separate
  problem is the transition to the $P$ framework, which is equivalent to the $F$ framework 
  only if $f(p)$ is a monotonic function, or only in such ranges of $f$ and $p$ in which 
  $f(p)$ is monotonic. There is also a technical problem of expressing $H$ as a function 
  of $p$ after its obtaining as a function of $f$ according to \rf{H}.
  
  For example, consider the simple rational function \cite{kru-20}
\beq                  \label{L-kru20}
		L(f) = \frac{f}{1 + 2\beta f}, \qq \beta = \const > 0.
\eeq  
  The quantity \rf{H} is easily found,
\beq                  \label{H-kru20}
		H = 2f L_f - L(f) = \frac{f (1- 2\beta f)}{(1 + 2\beta f)^2},
\eeq  
  but finding the dependence $f(p)$ to be substituted to \rf{H-kru20} requires solving
  a fourth-order algebraic equation:
\beq
		f = p (1+2\beta f)^4.
\eeq   
  
  It is therefore not surprising that the numerous existing electric solutions either start 
  from a specific function $H(p)$ or postulate the metric function $A(r)$, as is actually done 
  in \cite{gad-98} and a few other papers by the same authors. 

\medskip\noi  
  {\bf Dyonic solutions} with both nonzero charges $q_e$ and $q_m$ can be obtained with 
  more effort. Neither $f(r)$ nor $p(r)$ is known explicitly now. Thus, in particular,
\beq                                     \label{dy-f}
	 f(r) = \frac{2}{r^4} \biggl(q_m^2 - \frac{q_e^2}{L_f^2}\biggr). 
\eeq
  Comparing the expressions for $\rho(r)$ from \rf{SET2} and from \rf{A}, we can write
\beq    			      \label{dy-M'}
           \Half L(f) + \frac{2q_e^2}{L_f r^4} = \frac{2M'(r)}{r^2} = \rho(r). 
\eeq

  If $L(f)$ is known, \eqn{dy-f} can be treated either (A) as an (in general, transcendental) 
  equation for the function $f(r)$ or (B) as an expression of $r$ as a function of $f$. 

  In case (A), if we can find explicitly $f(r)$, integration of \eqn{dy-M'}
  gives the metric function $A(r)$.

  The scheme (B) gives a solution in quadratures expressed in terms of $f$ that can be now 
  chosen as a new radial coordinate. Indeed, if $L(f)$ and $r(f)$ are known 
  and monotonic, so that $L_f \ne 0$ and $r_f \ne 0$, we can rearrange \eqn{dy-M'} as 
\beq                                   \label{M_f}
               M_f = \frac{r^2 r_f}{2}\biggl[\frac{L}{2} + \frac{q_e^2}{L_f r^4}\biggr]
\eeq  
  (as before, the subscript $f$ denotes $d/df$). Since the r.h.s. of \rf{M_f} is known, 
  we can calculate $M(f)$ and $A(r)$ and also rewrite the metric in terms of the coordinate $f$.
  Thus we obtain {\it a general scheme of finding dyonic solutions} under the above conditions
  \cite{17-dyon}.   	`
  
  As a trivial example of using the scheme (A), we can consider the Maxwell theory, $L=f$. 
  Substituting $L=f$ and $L_f=1$ to \eqn{dy-M'}, we obtain $2M' = (q_e^2+q_m^2)/r^2$, whence 
  $2M(r) = 2m - (q_e^2+q_m^2)/r$ and 
\beq                    \label{dy-RN}
                     A(r) = 1 - \frac{2m}{r} + \frac{q_e^2+q_m^2}{r^2}, \ \ \ m=\const,
\eeq 
  that is, the dyonic Reissner-Nordstr\"om solution, as should be the case. 

  Another example is obtained \cite{17-dyon} if we assume that \eqn{dy-f} is linear in $f$.
  Then we have to put $L_f^{-2} = c_1 f + c_2$ with $c_{1,2} = \const$, which yields after 
  integration $L = L_0 + (2/c_1)\sqrt{c_1f + c_2}$. Assuming a Maxwell behavior, 
  $L \approx f$, at small $f$, we find $c_2 =1$, $L_0=-2/c_1$, and denoting $2/c_1 = b^2$, 
  we arrive at the truncated Born-Infeld Lagrangian,
\beq                  \label{L-BI}
   	L(f) = b^2 \Big(-1 + \sqrt{1+ 2 f/b^2}\Big), \ \ \ b = \const
\eeq
  (the full Born-Infeld Lagrangian also involves the other \emag\ invariant 
  $h^2 = (\sF\mn F\MN)^2$).
  With \rf{L-BI}, we obtain
\bearr              \label{f-BI}
            f(r) = \frac{2b^2 (q_m^2 - q_e^2)}{4q_e^2 + b^2 r^4},
\nnn                   \label{E-BI}
	   \rho (r) = -\frac{b^2}{2} + \biggl(\frac{b^2}{2} + \frac{2q_e^2}{r^4}\biggr)
			\sqrt{\frac{4q_m^2 + b^2 r^4}{4q_e^2 + b^2 r^4}}.
\ear  
  In the special case of a self-dual electromagnetic field, $q_e^2 = q_m^2$, we find simply
  $f=0$ and $\rho (r) = 2q^2/r^4$,  as in the Maxwell theory, and the dyonic 
  solution for $A(r)$ coincides with \rf{dy-RN}. For arbitrary charges, \eqn{A} leads to a
  long expression with the Appel hypergeometric function $F_1$, not to be presented here.
   
  Other examples of dyonic NED-GR solutions are found and discussed in 
  \cite{krugl-19a, krugl-19b, mkr-22, kru-20, yang-22}.
       
% ------------------------------------------------------    
\subsection{Regularity and no-go theorems} 
% ------------------------------------------------------    

{\bf Magnetic solutions.}    
  According to \rf{reg-c}, a regular center requires $A(r) = 1 + \cO(r^2)$ at small $r$. 
  In magnetic solutions with $f = 2q_m^2/r^4 \to \infty$ the metric regularity  
  then requires $L \to L_0 < \infty$ as $f\to\infty$ \cite{k-NED} because the density 
  that should be finite is now $T^t_t = \rho = L/2$. 
  Furthermore, asymptotic flatness requires $A(r) = 1 - 2m/r + o(1/r)$, where $m$ is the 
  Schwarzschild mass. By \rf{A}, it is the case if $\rho \sim r^{-4}$ or smaller as $r\to \infty$,
  which happens if $L(f) \sim f$, i.e., it has a Maxwell asymptotic behavior at small $f$.
  The metric is then approximately \RN\ at large $r$. 
  
  The infinite magnetic induction $B \sim 1/r^2$ at the center might cause a problem, 
  but as discussed in \cite{k-NED}, a correct estimate of the force applied to a charged test particle 
  moving in the nonlinear magnetic field under consideration, obtained along the lines of 
  Refs.\,\cite{rosen-52, ryb-book}, shows that such forces are finite for both electrically and 
  magnetically charged test particles and even vanish at $r=0$. 
  
  Thus invoking a smooth function $L(f)$ such that $L \sim f$ as $f\to 0$ and $L \to L_0 < \infty$
  as $f\to\infty$ is an easy way to obtain globally regular configurations including magnetic \bhs\
  and solitons, used in many papers, probably beginning with Ref. \cite{k-NED}.     

  In all such solutions, a general feature is that $A\to 1$ as both $r\to 0$ and $r\to\infty$.
  Moreover, the mass term $-2m/r$ contributes negatively to $A(r)$ as long as $m > 0$.
  Thus in regular solutions $A(r)$ should inevitably have a minimum, at which the value of 
  $A$ depends on the mass and charge values. Their relationship determines the existence
  of horizons located at regular zeros of $A(r)$. If the mass $m > 0$ is fixed, then at small charges
  (which contribute positively to $A(r)$ at least at large $r$) the minimum of $A$ is negative 
  because the solution is close to Schwarzschild's almost everywhere, and then any regular 
  function $A(r)$ has two zeros, one of which should be close to $r=2m$, while the other emerges 
  since it is necessary to return to $A(r) > 0$ at small $r$ to reach $A=1$ at $r=0$. At large 
  charges, on the contrary, the mass term $-2m/r$ is only significant at large $r$, and a minimum 
  of $A$ should be positive, leading to a solitonic solution. Some value of $q$ must be critical,
  leading to a double zero of $A(r)$, corresponding to a single extremal horizon. 
  
  This general picture is really observed in the known examples of regular \ssph\ NED-GR solutions. 
  Let us illustrate it with the behavior of $A(r)$ in the example from \cite{k-NED}, where
\beq                              \label{L-ch}
				L(f) = \frac{f}{\cosh^2 \big(b |f/2|^{1/4}\big)}, \qq b = \const > 0.
\eeq  
   In the magnetic solution, with $q= q_m >0$ (for simplicity),
\beq                             \label{A-ch}
				\rho = \frac {q^2/r^4}{\cosh^2 (b\sqrt{q}/r)}, \qq
				A(r) = 1 - \frac{2m}{r} \bigg(1 - \tanh \frac{q^2}{2mr}\bigg),
\eeq   
  where the mass $m$ is determined as $M(\infty)$.
  The behavior of $A(r)$ is shown in Fig.\,1 for three values of $q/m$ leading to qualitatively different
  geometries. The causal structures and Carter-Penrose diagrams of these space-times are the same 
  as those for \RN\ ones, but the important difference is that now the lines $r=0$ denote 
  a regular center instead of a singularity. 
  
  Regular models with more than two horizons are also possible, see, e.g., \cite{od-17, gao-21} 
  for detailed studies of such solutions.
  
% --------------------------- fig 1
\begin{figure}
\sidecaption
%\centering
\includegraphics[width=7cm]{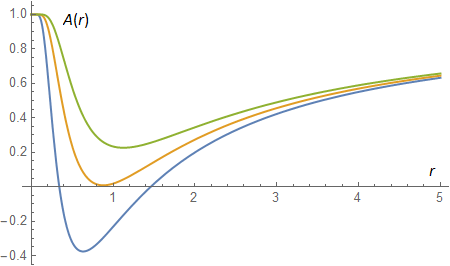} 
\caption{The behavior of $A(r)$ according to \eqn {A-ch} with $m =1$ and 
		$q =0.9,\ 1.06,\ 1.2$ (bottom-up). }
		\label{fig1}
\end{figure}  
% ---------------------------- 

  An important feature of regular solutions is that with given $L(f)$ the \Scw\ mass $m$ is uniquely 
  fixed by the charge $q$. Indeed, to obtain a regular center, the integration in \eqn{A} must be 
  carried out from $r=0$ (where the density $\rho$ is finite and determined by $q$) to arbitrary $r$, 
  resulting in the \Scw\ mass $m = M(\infty)$. It means that this mass is completely created by the 
  \emag\ field energy. Any additional mass $m_1$ that can appear in the solution as an integration 
  constant in \eqn{A} would add the singular term $2m_1/r$ to $A(r)$.
  
  Thus, in particular, returning to the solution \rf{A-ch} for the theory \rf{L-ch}, it is easy to find 
  that $m = q^{3/2}/(2b^{1/4})$, or on the contrary, the parameter $b$ in $L(f)$ may be 
  expressed in terms of $m$ and $q$: $b = q^6/(16 m^4)$.  
  
\medskip\noi
{\bf Electric solutions} with a regular center and a \RN\ asymptotic behavior can either be found in
  the same manner using the $P$ formulation of NED (as is done in Refs. \cite{gad-98, FW-16} 
  and many others), or obtained directly from the magnetic ones using the FP duality. 
  However, as we saw above, solutions with the same metric correspond to quite different NED 
  theories than those used in magnetic solutions, and this circumstance leads to their different
  physical properties. First of all, let us recall a theorem proved in \cite{B-Shi, B-Shi2, k-NED}:
%  =========  
\begin{theorem}    \label{theo1}
  If a \ssph\ electric solution ($q_e\ne 0, q_m =0$) to the $L(f)$ NED-Einstein equation describes 
  a space-time with a regular center, it cannot have a Maxwell behavior at small $f$
  ($L \approx f, \ L_f \to 1$).
\end{theorem}    

\begin{proof}  
  To begin with, since the Ricci tensor for our metric is diagonal, the curvature
  invariant $R\mn R\MN = R\mN R\nM$ is a sum of squares of the components $R\mN$, 
  hence each of them taken separately must be finite at any regular point, including a center. 
  It then follows that each of the components of $T\mN$ should be finite, as well as their 
  any linear combination. In particular, by \rf{SET2}, we must have $|f_e L_f|  < \infty$.
  But according to \rf{ff}, $f_e L_f^2 = 2q_e^2/r^4 \to \infty$. These two conditions, 
  taken together, lead to
\beq                   \label{L_fe}
			f = -f_e \to 0, \qq   L_f \to \infty \qq {\rm as}\ \ r\to 0.
\eeq    
  It means that we have a non-Maxwell function $L(f)$ at small $f$.
\qed  
\end{proof}  
% ==========
  
  On the other hand, regular \asflat\ electric solutions obtained in the $P$ formulation of NED
  have a correct Maxwell asymptotic behavior. How can it be combined with \rf{L_fe}?
  
  The answer is that such solutions correspond to different Lagrangians $L(f)$ near $r =0$ 
  and at large $r$ \cite{k-NED, B-comment}. Indeed, at a regular center $r=0$ we have
  $-p = 2q_e^2/r^4 \to \infty$ and $f =0$, while at flat infinity both $p\to 0$ and again $f\to 0$. 
  It means that $f$ inevitably has at least one extremum at some $p=p^*$, breaking the 
  monotonicity of $f(p)$, which means that on different sides of $p^*$we have different 
  functions $L(f)$ corresponding to the same $H(p)$. As shown in \cite{k-NED}, at an extremum 
  of $f(p)$ the function $L(f)$ suffers branching, at which the derivative $L_f$ tends to the same 
  finite limit as $p\to p^*+0$ and $p\to p^*-0$, while $L_{ff}$ tends to infinities of opposite signs. 
  This corresponds to a cusp in the plot of $L(f)$. Another form of branching of $L(f)$ takes
  place at extremum points of $H(p)$, if any, where the monotonicity of $f(p)$ also breaks 
  down. The number of different Lagrangians $L(f)$ on the way from the center to infinity 
  is equal to the number of monotonicity ranges of $f(p)$ \cite{k-NED}. 
  
  To illustrate this unusual behavior of $L(f)$ let us use as an example the same metric function 
  \rf{A-ch}, where now $q = q_e$, as a solution corresponding to $H(p)$ dual to \rf{L-ch}:
\beq                              \label{H-ch}
				H(p) = -\frac{p}{\cosh^2 \big(b |p/2|^{1/4}\big)}, \qq b = \const > 0.
\eeq    
  Calculations reveal the behavior of the corresponding functions $f(p)$ and $L(f)$ 
  shown in Figs.\,2, 3. It turns out that $L(f)$ has as many as four branches, in other words,
  there are four NED theories acting in different parts of space. 
  
% --------------------------- fig 2
\begin{figure}
\centering
\includegraphics[width=110mm]{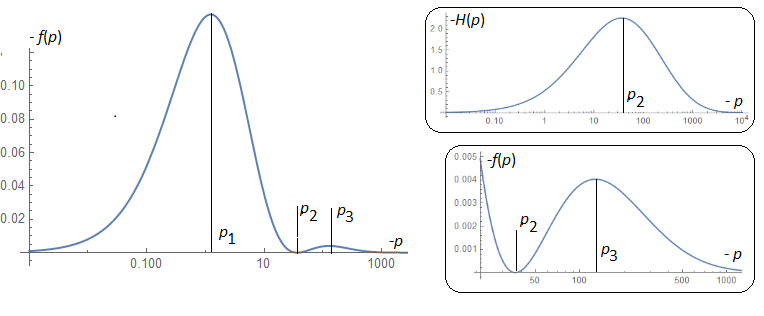}
\caption{\small
		The function $f(p)$ obtained from $H(p)$ given by  \eqn {H-ch} with $b =1$. 
		The points $p_1$ and $p_3$ show the maxima of $|f(p)|$ while $p_2$ shows its
		minimum corresponding to the maximum of $|H(p)|$. The upper inset shows 
		the function $H(p)$, while the lower one is an enlarged view of the neighborhood 
		of $p_2$ and $p_3$ in the plot of $f(p)$.}
		\label{fig2}
\end{figure}  
% --------------------------- fig 3
\begin{figure}
\centering
\includegraphics[width=10cm]{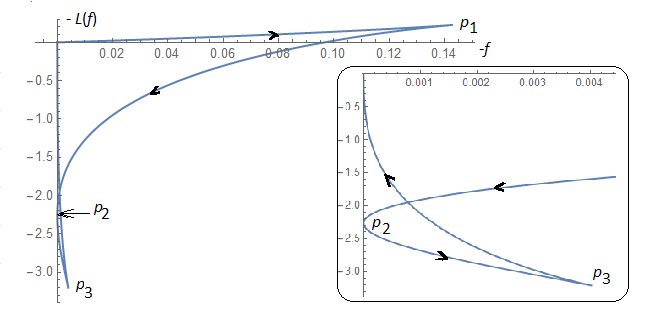} 
\caption{\small
		The behavior of $L(f)$ in the electric solution for $H(p)$ from \eqn {H-ch} 
		with $b =1$. The points $p_1, p_2, p_3$ correspond to the extrema of $f(p)$,
		at which the function $L(f)$ passes on from one branch to another.
		The inset shows more clearly the range close to $p_2$ and $p_3$.
		Arrows on the curves show the direction of growing $|p|$.  }
		\label{fig3}
\end{figure}  
% ----------------------------   
  
\medskip\noi
{\bf Dyonic solutions.} If there are both nonzero $q_e$ and $q_m$, then a 
   combination of $T\mN$ components leads to the requirement
\beq
			(f_e + f_m) |L_f| < \infty
\eeq     
  that must hold at any regular point, including a regular center. Moreover, it must 
  hold for each term separately because both $f_e$ and $f_m$ are positive. 
  Applying it to $f_e$, we obtain, as before, that $L_f  \to \infty$ at a regular 
  center \cite{k-NED}.  However, the inequality $f_m|L_f| < \infty$ leads to 
  the requirement $L_f \to 0$, since $f_m = 2q_m^2/r^4 \to \infty$ as $r\to 0$.
  We arrive at a contradiction that leads to the general result:
% =========  
\begin{theorem}      \label{theo2}
  {\sl Static \sph\ dyonic solutions to the NED-Einstein equations with arbitrary $L(f)$ 
  cannot describe space-times with a regular center.}
\end {theorem}  
% =========

\noi  
{\bf Inclusion of a cosmological constant.}
  If $\Lambda \ne 0$, asymptotically (A)dS solutions \cite{Mat-09} 
  are obtained by simply adding $ - \Lambda r^2/3$ to $A(r)$ in \rf{A}. 
  This new term does not affect the properties of the solutions near $r=0$,
  therefore, all conclusions on the existence of a regular center and the 
  necessary conditions for it, obtained with $\Lambda =0$, remain valid  
  with $\Lambda \ne 0$, although the latter drastically changes the global 
  properties of space-time.
  
% ---------------------------------------------  
\subsection{Causality and unitarity}
% ---------------------------------------------
  
   An important viability criterion for NED theories has been suggested by A. Shabad and V. Usov 
   \cite{usov-11}, partly on the basis of their previous work: they have used (i) the causality principle as 
   the requirement that elementary excitations over a background field should not 
   have a group velocity exceeding the speed of light in vacuum and (ii) the unitarity principle formulated 
   as the requirement that the residue of the propagator should not be negative. As a result, there emerge 
   the following inequalities that should hold for a theory satisfying these principles: in our notations,
   for $L(f)$ theories,
\beq                             \label{causal}
		L_f > 0, \qq      L_{ff} \leq 0, \qq         \Phi =L_f + 2f L_{ff} \geq 0. 
\eeq    
  One can notice that the third condition can be rewritten as $H_f \geq 0$, with 
  the Hamiltonian-like quantity $H$ given by \rf{H}, but does not directly concern the derivative 
  $H_p$ due to a possible complexity in the dependence $f(p)$. The quantity $\Phi$ also plays
  an important role in the effective metric for photon propagation and in the stability conditions, 
  to be considered in the next subsections.

  One immediate observation can be made about magnetic solutions with a regular center, 
  both \bh\ and solitonic ones:
% =========  
\begin{theorem}   \label{theo3}
       In \ssph\ magnetic solutions to $L(f)$ NED-Einstein equations, the causality and unitarity 
       conditions \rf{causal} are inevitably violated in a neighborhood of a regular center.
\end{theorem}   
\begin{proof}  
  A regular center requires a finite limit of $L(f)$ as $f\to\infty$. If $L_f > 0$ (as required by the 
  first inequality in \rf{causal}), then, to have a convergent integral $L = \int L_f df$, one has to 
  require $L_f \ll 1/f$ at large $f$, hence the quantity $L_f \sqrt{f}$ is decreasing as $f\to \infty$,
  and its derivative in $f$ is negative. On the other hand, we can write 
  $\Phi = 2\sqrt{f}(L_f\sqrt{f})_f$, consequently, $\Phi < 0$ at large $f$, so that the
  first and third inequalities in \rf{causal} cannot hold simultaneously. 
\qed
\end{proof}
% =========

% -------------------------------------------------------------
\subsection{Light propagation and the effective metric}
% -------------------------------------------------------------

  As we have seen, the same regular metric of the form \rf{ds2} can be obtained 
  with two kinds of sources, the electric and magnetic ones, described by different NED
  theories. It is thus natural to expect that the properties of \emag\ fields will also be
  different in these two cases. Let us try to explore these differences using the effective 
  metric formalism developed by M. Novello et al. \cite{nov-1} while studying the 
  propagation of electromagnetic field discontinuities using Hadamard's approach
  \cite{hadamard}. According to \cite{nov-1, nov-2}, photons governed by NED 
  propagate along null geodesics of the effective metric
\beq                                                           \label{geff}
        h\MN =  g\MN  L_f - 4 L_{ff} F^\mu{}_\alpha F^{\alpha\nu}.
\eeq

\noi
{\bf Electric solutions.}
  In the case of a purely electric field in the metric \rf{ds2}, $h\MN$ is diagonal, and we
  can write the effective metric as
\bearr      \label{ds-ee}
        ds^2_{\rm eff} \equiv  h\mn dx^\mu dx^\nu
	 = \frac{1}{\Phi} \bigg[ A(r) dt^2 - \frac{dr^2}{A(r)} \bigg]
			 - \frac{r^2}{ L_f}d\Omega^2, 
\nnn
     	   \Phi = L_f + 2f L_{ff} = \frac{H_p}{f_p}.
\ear

  Consider the behavior of $h\mn$ at branching points of $L(f)$ that are inevitable in 
  solutions with regular $g\mn$. At an extremum $p=p^*$ of $f(p)$ at which $f\ne 0$ 
  (like points $p_1$ and $p_3$ in Fig.\,2), we have $\Phi \to \infty$ since $f_p =  0$ while
  $H_p$ is finite. This results in a curvature singularity of the effective metric
  due to blowing up of the quantity $K_1$ in \rf{K_i}.

  Another kind of singularity of the metric \rf{ds-ee} occurs at extrema of $H(P)$,
  those like point $p_2$ in Fig.\,2: in this case, generically, $\Phi$ is finite but
  $L_f \to \infty$, which leads to a singular center in the auxiliary space-time 
  with the metric $h\mn$ due to $h_{\theta\theta}\to 0$.
 
  The changes in photon frequencies at their motion in space-time can be evaluated 
  as outlined in \cite{nov-2}. Thus, if an emitter at rest at point $X$ sends a photon 
  with frequency $\nu_X$, it comes to a receiver at rest at point $Y$ with frequency
  $\nu_Y$ related to $\nu_X$ by
\beq               \label{z}
    \frac{\nu_Y}{\nu_X} = \bigg[\frac{\sqrt{g_{tt}}}{h_{tt}}\bigg]_Y 
						\bigg[\frac{\sqrt{g_{tt}}}{h_{tt}}\bigg]_X^{-1}
        	= \bigg[\frac{\Phi}{\sqrt{A}}\bigg]_Y 
			   \bigg[\frac{\Phi}{\sqrt{A}}\bigg]_X^{-1},
\eeq
  where the second equality sign corresponds to the metric (\ref{ds-ee}).
  If $X$ is a regular point while $Y$ is located at an inevitable branching point of 
  $L(f)$ (like $p_1$ or $p_3$), then any photon arriving there is infinitely blueshifted,
  gaining an unlimited energy, which thus implies instability of the whole configuration.
  
  The above reasoning used the assumption $A >0$. In \bh\ solutions, the sphere 
  where $L_f =0$ may be located beyond the event horizon, where $A < 0$.
  In such a region, also called a T-region, $r$ is a temporal coordinate, $t$ is a spatial
  one, and in the redshift relation \rf{z} we must replace $g_{tt}$ with $g_{rr}$, or
  more specifically, $\sqrt{A}$ with $1/\sqrt{-A}$. However, as long as $A$ is finite,
  this replacement does not affect the conclusion on an infinite blueshift on the sphere    
  where $\Phi = \infty$.
  
\medskip\noi
{\bf Magnetic solutions.}
   For the same metric $g\mn$ with a magnetic source, we get, instead of (\ref{ds-ee}), 
\beq
        ds^2_{\rm eff} =                                     \label{ds-em}
    \frac{1}{ L_f}\biggl[ A(r) dt^2 - \frac{dr^2}{A(r)} \biggr] - \frac{r^2}{\Phi}d\Omega^2,
\eeq
  where, as before, $\Phi  =  L_f + 2f L_{ff}$. Then, for a photon traveling from 
  point $X$ to point $Y$, we find instead of \rf{z}:
\beq
    \frac{\nu_Y}{\nu_X}                                    	\label{zm}
        = \bigg[\frac{ L_f}{\sqrt{A}}\bigg]_Y \bigg[\frac{L_f}{\sqrt{A}}\bigg]_X^{-1}.
\eeq
  Now $L(f)$ has no branching points, while at a regular center ($r=0,\ A=1$) 
  both $L_F$ and $\Phi$ vanish, and the quantity $h_{22}\to \infty$, i.e., the 
  spherical radius in the effective metric behaves as if in a wormhole, 
  whereas $h_{tt} \to \infty$, which means an infinite redshift for photons. 
  Also, all curvature invariants of the metric (\ref{ds-em}) vanish at $r=0$. 
  It is really a quiet place.

  There still occurs something of interest between spatial infinity and the center 
  of a regular magnetic model: there is necessarily a sphere $r=r^*$ on which $\Phi=0$. 
  Indeed, $\Phi$ can be presented as $\Phi = 2\sqrt{f} (\sqrt{f} L_f)_f$.  
  The quantity $\sqrt{f} L_f$ tends to zero both at $r=0$ (where $f\to \infty$
  but $L\to\const$) and in the limit $r\to\infty$. Since $\sqrt{f} L_f$ is in general 
  nonzero, it has at least one extremum at some $f \ne 0$, thus it is the value where 
  $\Phi=0$. The metric (\ref{ds-em}) is singular there due to $h_{22}\to\infty$, 
  but this singularity seems to be unnoticed by the photons, as follows from 
  an integral of their geodesic equation
\beq
     	L_f^{-2} \dot r{}^2  + [A(r)\Phi/r^2] \ell^2 = \ep^2,   \label{geo}
\eeq
  where the overdot denotes a derivative in an affine parameter, 
  $\ep$ and $\ell$ are the photon's constants of motion characterizing
  its initial energy and angular momentum. Generically we have $L_f\ne 0$ at points 
  where $\Phi=0$, therefore the photon frequency remains finite. However, as we 
  will see below, the photon velocities behave there in an unusual manner.
  
  If $ L_f=0$ at some value of $f > 0$, it leads to another kind of singularity of 
  the metric \rf{ds-em}), and this time it acts for NED photons as a potential wall, or 
  a mirror, as is evident from (\ref{geo}) which then implies $\dot r =0$. Also, 
  from \eqn{zm} it follows that the photons are infinitely redshifted there: $\nu_Y$ 
  vanishes if $ L_f (Y)=0$. It means that in such a case no photon from outside can 
  approach the center.

  We thus observe a striking difference between the properties of photons moving
  in the same regular metric \rf{ds2} in the cases where it is sourced by electric
  and magnetic fields. In the electric case, the photons inevitably ``accelerate'' 
  to an infinite energy and destabilize the whole system, whereas in the magnetic 
  case, even if they can approach the regular center (if $L_f \ne 0$, hence no mirror),
  they lose energy, being infinitely redshifted there.
  
  The violent behavior of photons in electric regular \bhs\ was discovered by
  Novello et al. \cite{nov-2} for a particular example of such a configuration. 
  As shown in \cite{k-NED}, it is quite a general property of NED-GR solutions. 
  
\medskip\noi
{\bf Photon velocities.}  
  A question of interest is the velocity of NED photons in regular or singular space-times.
  From \rf{ds-em} it follows that radially moving photons have the same velocity  
  equal to $c$ (=1) as the conventional Maxwell ones since the 2D metric of the ($t, r$)
  subspace in the effective metric \rf{ds-em} is conformal to that in the space-time 
  metric \rf{ds2}, and their 1D light cones coincide. This is true for both electric 
  and magnetic solutions. However, the situation is different for nonradial photon paths. 
  
  Consider a photon moving instantaneously in a tangential direction. Without loss of 
  generality we can suppose that it moves along an equator of certain radius 
  $r$ in our coordinate system. For the corresponding null direction in terms of the
  effective metric we have $ds^2_{\rm eff} = h_{tt}dt^2 - h_{\theta\theta} d\theta^2 =0$, and
  for the photon's linear velocity $v_{\rm ph} = r \,d\theta/dt$ we obtain:
\bear            \label{v_ph}
	\text{in an electric solution:}	&&    v_{\rm ph}^2 = {A L_f}/{\Phi},
\nnv
	\text{in a magnetic solution:}	&&    v_{\rm ph}^2 = {A \Phi}/{L_f},
\ear
  For the Maxwell field, $L_f \equiv \Phi \equiv 1$, hence $ v_{\rm ph}^2 = A$, and it would be
  equal to unity if we used the local time increment $dt_{\rm local} = \sqrt{A} dt$ instead of the 
  coordinate time increment $dt$, and the length element equal to $dr/\sqrt{A}$ instead of $dr$. 
  Thus, as should be the case, Maxwell photons always travel in vacuum with the speed of light. 
  The factor $L_f/\Phi$ or $\Phi/L_f$ changes the photons' velocity, working like a refractive index. 
  
  In particular, in electric solutions, at cusplike branching points where $\Phi \to \infty$ 
  while $L_f$ remains finite (like points $p_1$ and $p_3$ in Figs.\,\ref{fig2} and \ref{fig3}), 
  $v_{\rm ph} \to 0$, in other words, tangentially moving photons have zero velocity
  at this value of $r$. On the contrary, at branching points like $p_2$, where $H_p = f_p = 0$ 
  and $\Phi$ is finite but $L_f \to \infty$, we obtain $v_{\rm ph} \to \infty$.  

  In magnetic solutions, at spheres where $\Phi =0$ while $L_f$ is finite, we have
  again $v_{\rm ph} =0$, a zero velocity of tangentially moving photons. 

  So far we were assuming $A(r) > 0$, while in \bh\ space-times there are T-regions where  
  $A(r)$ is negative. However, the only change in \eqn{v_ph} emerging in a T-region 
  is the simple replacement $A \to 1/|A|$ because $r$ is there a time coordinate instead 
  of $t$, and in other respects our reasoning remains unaltered. 

  At intermediate directions between the radial and tangential ones, the NED photon 
  velocities will obviously have intermediate values. We conclude altogether that these velocities 
  can be both subluminal and superluminal, varying from zero to infinity. 
  
  We also observe that the conditions under which superluminal photon velocities are avoided
  ($L_f/\Phi \leq 1$ for electric solutions and $\Phi/L_f \leq 1$ for magnetic ones) do not coincide 
  with the causality/unitarity conditions \rf{causal}. Even more than that: any non-Maxwell
  NED, in which $\Phi/L_f \not\equiv 1$, predicts superluminal photon motion in either 
  electric or magnetic space-times.  Actually, this observation puts to doubt either any NED
  theory or the described straightforward interpretation of the effective metrics. 

  Also, the nonlinearity of NED is acting like a highly anisotropic medium, which, if one takes 
  into account the wave properties of photons, naturally leads to such a phenomenon 
  as birefringence, see the relevant recent studies in \cite{krug-15, sangpyo-22} and
  references therein. 
  
% ---------------------------------------------  
\subsection{Dynamic stability}
% ---------------------------------------------
 
 Any static or stationary configuration may be regarded viable if it is stable under different kinds
 of perturbations, which always exist in nature, or at least if it decays slowly enough. Possible regular 
 NED \bhs\ do not make an exception, and their stability is discussed in a number of 
 papers, e.g., \cite{Mor-03, breton-14, tosh-19, soda-20}, see also references therein.
 
 C. Moreno and O. Sarbach \cite{Mor-03} have derived sufficient conditions for linear dynamic stability
 of the domain of outer communication of electric or magnetic \bhs\ sourced by a general $L(f)$ NED.
 For magnetic \bhs\ these conditions read (in the present notations)
\bearr                  		\label{stab1}
		L>0, \qq  L_y >0, \qq  L_{yy} >0,  
\yyy                  		\label{stab2}
		3 L_y - A(r) y L_{yy} \geq 0, 		
\ear 
 where $y: = \sqrt{q^2 f/2} = q^2/r^2$, and the index ``$y$'' stands for $d/dy$. In terms of 
 $f$ these conditions are rewritten as
\bearr                  		\label{stab3}
		L>0, \qq  L_f >0, \qq  \Phi \equiv L_f + 2f L_{ff} >0,  
\yyy                  		\label{stab4}
		[6 - A(r)] L_f - 2 f L_{ff} \geq 0. 		
\ear 
  One can notice that the conditions \rf{stab3} partly coincide with the causality and unitarity 
  conditions \rf{causal}. Moreover, if $L_f > 0$ and also the condition $L_{ff} \leq 0$ from 
  \rf{causal} is valid, then the condition \rf{stab4} holds automatically provided $f >0$ (which 
  is true for magnetic solutions) and $A(r) < 6$ (we can note that at least in regular \bh\ solutions, 
  in general, $A(r) \leq 1$).
   
  Thus \eqn{stab4} is not expected to make a problem, at least for regular magnetic solutions. 
  Unlike that, by Theorem \ref{theo3}, the condition $\Phi > 0$ is always violated for such solutions
  near a regular center. This may be important for \bh\ solutions only if the range of $r$ where 
  $\Phi < 0$ extends to the domain of outer communication, which must be checked for each 
  particular \bh\ solution. 
  
  The sufficient stability conditions for electric \bhs\ have a form similar to \rf{stab1}, \rf{stab2}
  in terms of the P-framework of the theory \cite{Mor-03}, which could be expected due to FP 
  duality. Thei reformulation to the F-framework is not possible in a general form due to problems
  with a relationship between $f$ and $p$, see above.     
    
  More general stability conditions for NED-GR solutions involving both $F\mn$ and $\sF\mn$ 
  have been recently obtained by K. Nomura, D. Yoshida and J. Soda in \cite{soda-20}.
  
% --------------------------------------
\subsection{Examples}
% --------------------------------------

  Let us enumerate some particular examples of the Lagrangians $L(f)$ discussed in the 
  literature, along with their basic properties at $f >0$ (that is, for their magnetic solutions): 
  the existence of a correct Maxwell weak field (MWF) limit, a finite limit as $f\to\infty$, 
  necessary for a regular center in magnetic solutions, and the validity of the causality, 
  unitarity and stability conditions, \rf{causal} and \rf{stab3}.

% ------------------------------------------------- tab 1  
\begin{table}
%\centering
\caption{\small            \label{tab1}     
		Some examples of $L(f)$ NED theories: properties of magnetic solutions ($f\geq 0$)  }
\begin{tabular}{p{1.8cm}p{3.5cm}p{1.5cm}p{1.5cm}p{1.5cm}p{1.5cm}}
\hline\noalign{\smallskip}
   References  & Lagrangian$^a$ 
   & Correct MWF limit & Finite\,\,as\qq $f\to\infty$ &  Condition $L_{ff} < 0$ & Condition $\Phi >0$\\
\noalign{\smallskip}\svhline\noalign{\smallskip}
    \cite{B-Inf}, \rf{L-BI} & $\beta^2 \big(\!-1 +\! \sqrt{1+2f/\beta^2}\big)$  &  yes  & no  & yes  & yes
\yy
    \cite{Kru-16, kru-21}, \rf{L-kru20} & $\dfrac{f}{1 + 2\beta f}$   &  yes   & yes  & yes  & partly$^b$
\\[10pt]
    \cite{kru-16a}  & $\beta^{-1} \arctan (\beta f)$                    & yes   & yes   & yes  & partly
\yy
	\cite{kru-16b} &  $\beta^{-1} \arcsin (\beta f)$                   & yes   & no    &  no   &  yes
\yy
	\cite{kru-19}  &  $\beta^2 \log \Big( 1 + \dfrac {f}{\beta^2}\Big)$ & yes  &  no  & yes  & partly
\yy
    \cite{FW-16}   &   $\dfrac{f}{(1 + (\beta f)^{1/4})^4}$ & yes & yes  & yes  & partly
\\[10pt]
	\cite{k-NED}, \rf{L-ch}	& $\dfrac{f}{\cosh^2(\beta |f/2|^{1/4})}$   & yes   & yes  & no & partly
\\				  					
\noalign{\smallskip}\hline\noalign{\smallskip}
\end{tabular}
$^a$ In all examples, $\beta = \const >0$. \\
$^b$ Here and in other lines, ``partly'' means that $\Phi >0$ at $f$ smaller than some critical value.
\end{table}
% ---------------------------------------

  It is convenient to do that in the form of a table, see Table 1. Among the conditions 
  \rf{causal} and \rf{stab3} we select there the inequalities $L_{ff} < 0$ and $\Phi > 0$
  because the condition $L > 0$ holds in all examples, and $L_f > 0$ in all of them except 
  the one with hyperbolic cosine. 

  The first line represents the truncated Born-Infeld Lagrangian which does not provide a regular
  center but satisfies the conditions \rf{causal} and \rf{stab3}. The next four lines correspond to
  different examples of NED considered by S. Kruglov, the first two of them provide regular 
  magnetic \bhs. The sixth line represents a special case from numerous examples considered 
  by Fan and Wang in \cite{FW-16}, selected there because it both has a correct MWF limit 
  and provides a regular center. The last line is the special case of NED discussed above.   
  The explicit form of the solutions can be found in the cited papers along with detailed 
  discussions of their properties. This list certainly does not pretend to be complete, and many 
  other solutions have been obtained and studied.   
  
  It can be observed from the table that in all NED theories that provide a regular center 
  (those with ``yes'' in the column ``finite as $f\to\infty$''), the inequality $\Phi > 0$ does 
  not hold at sufficiently high values of $f$, in accordance with Theorem \ref{theo3}.

% =====================================
\section{NED with more general Lagrangians}\label{sec5}
% =====================================
\subsection{Systems with $L = L(f,h)$}
% -----------------------------------------------------

  Beginning with the paper by Born and Infeld \cite{B-Inf}, the researchers considered NED theories 
  with Lagrangians more general than $L(f)$, depending on \emag\ invariants other than $f$. 
  The first and the most natural candidate is the pseudoscalar $h = \sF\mn F\MN = 2{\bf BE}$, 
  where {\bf E} and {\bf B} are the the electric field strength and magnetic induction 3-vectors, 
  respectively. Now the total action has the form
\beq            \label{S2}
			S = \Half \int \sqrt {-g}\, d^4 x [R - L (f,h)]. 	
\eeq
  Special cases of $L(f,h)$ are the Born-Infeld Lagrangian
\beq              \label{BInf}
			L^{\rm BI} =  b^2 \bigg(-1+ \sqrt{1 + \frac f{2b^2} - \frac{h^2}{16 b^4}}\bigg),
			\qq b >0,
\eeq  
  and the so-called modified Maxwell (ModMax) Lagrangian \cite{bandos-20, sorokin-21}
\beq              \label{ModMax}
			L^{\rm MM} = \frac 14 \Big( f \cosh \gamma - \sqrt{f^2 + h^2} \sinh\gamma \Big),
			\qq  \gamma \in \R.
\eeq  
  Both these models are distinguished by their symmetry properties, in particular,
  the ModMax NED is conformally and duality invariant.\footnote
  		{In our notations, see \rf{S2}, some of the signs and factors are different from those 
  		 in \cite{bokulic-22} and other papers. In particular, the Maxwell theory here 
  		 corresponds to $L(f, h) = f$.}
  		 
  The \emag\ field equations due to \rf{S2} read
\beq              \label{Max*}
		\nabla_\mu(L_f F\MN - L_h \sF\MN) =0, \qq    \nabla_\mu \sF\MN =0,
\eeq  
  and the \emag\ field SET has the form   
\beq
             T\mN = -2 L_f F_{\mu\alpha} F^{\nu\alpha} + \half \delta\mN (L -h L_h).
\eeq  

  Assuming static spherical symmetry, hence having only radial electric and magnetic fields, 
  we are again dealing with a SET with $T^r_t =0$ and $T^t_t = T^r_r$, and the metric 
  can be written in the form \rf{ds2}. We then have according to \rf{Max*}
\beq
		L_f F^{tr} - L_h \sF^{tr} = \frac{q_e}{r^2}, \qq    F_{\theta\phi} = q_m \sin \theta,
\eeq  
  with the corresponding charges $q_e, q_m = \const$.
  
  For \ssph\ solutions to the NED-GR equations with $L = L(f,h)$, a number of no-go theorems 
  have been proved in Ref.\,\cite{bokulic-22}. According to these theorems, such solutions 
  with the metric \rf{ds2} cannot describe a geometry with a regular center under the 
  following assumptions on the \emag\ field:
\begin{enumerate} \itemsep 5pt
\item %1
	$q_e \ne 0$, $q_m =0$ (electric), MWF limit.
\item %2
  	$L(f, h) = L(f)$, $q_e \ne 0$, $q_m \ne 0$ (dyonic).
\item %3
	$L(f, h) = f + \eta(h)$, with an arbitrary function $\eta(h)$, $q_m \ne 0$ (magnetic or dyonic).
\item %4
	$L(f, h) = f + a f^s h^u$, with real $a \ne 0$, positive integers $s > 1,\ u >1$, and
	$q_e \ne 0$, $q_m \ne 0$ (dyonic).
\item %5  
     $L(f, h) = f + af^2 + bfh + ch^2$, where $a, b, c \in \R$, $q_e \ne 0$, $q_m \ne 0$ (dyonic).    
\item %6
	$L(f, h)$ given by \rf{BInf} or \rf{ModMax}, $q_m \ne 0$ (magnetic or dyonic).             
\item %7
	$L(f, h) = f + af^2 + bfh + ch^2$, the pair $(b,c) \ne (0,0)$,
								$q_e = 0$, $q_m \ne 0$ (magnetic).
\end{enumerate}
  The numbering here corresponds to the theorem numbers in Ref.\,\cite{bokulic-22}.
  Theorem 2 from this list coincides with our Theorem \ref{theo2} presented in Section \ref{sec4}.
  We can notice that only two theorems, the first and the sixth ones, use the assumption
  of a correct MWF limit: in all other cases considered, a regular center is impossible irrespective of
  the weak field behavior of the theory.    
  
  On the other hand, there still remain some opportunities of obtaining regular BHs other 
  than purely magnetic ones with $L = L(f)$. For example, both with $L(f)$ and $L(f,h)$,
  purely electric solutions with a regular center are possible with a theory having no MWF limit.
  Then, assuming a regular central region governed by such a theory, one can obtain an \asflat\ 
  electrically charged configuration by using a kind of phase transition, such that outside a 
  certain sphere $r = r_{\rm crit}$, another NED theory will be valid, having a correct MWF
  limit, as was suggested in \cite{Bur1}. 
   
% -----------------------------------------------------  
\subsection{Systems with $L = L(f, J)$}  
% -----------------------------------------------------
  
  One more invariant, in addition to $f$, 
\beq                   \label{J4}
		J \equiv J_4 = F\mn F^{\nu\rho} F_{\rho\sigma} F^{\sigma\mu},
\eeq      
  has also been used for formulating an extended NED theory \cite{denis-19, sokolov-21}.
  With this invariant, the action reads
\beq            \label{S4}
			S = \Half \int \sqrt {-g} d^4 x [R - L (f,J)], 	
\eeq  
  the \emag\ field equations are
\bearr                \label{Max4}
			\nabla_\mu Q\MN =0, \qq    \nabla_\mu \sF\MN =0,
\yyy
			Q\MN := 4L_f F\MN + 8 L_J F^{\mu\rho} F_{\rho\sigma} F^{\sigma\nu},			
\ear  
  where $L_f = \d L/\d  f$ and $L_J = \d L/\d J$. The SET has the form
\beq                                   \label{SET-J}
			T\mN = -2 (L_f +f L_J)F_{\mu\alpha} F^{\nu\alpha} 
						+ \Half \delta\mN [(f^2 - 2J)L_J + L]. 
\eeq  
  An important subclass of the theories \rf{S4} is called conformal NED, or CNED, and is 
  characterized by a zero trace of the SET \rf{SET-J} \cite{denis-19, sokolov-21}, hence,
\beq                       \label{CNED}
	               T = 	2(L - f L_f -2J L_J) =0. 	
\eeq  
  In this case, the SET as a whole is a multiple of the Maxwell field SET, and the 
  field equations \rf{Max4} are invariant under general conformal mappings of the metric 
  (conformally invariant) like the Maxwell equations. 
  
  The theory has a Maxwell asymptotic behavior at small fields if $L(f, J) \approx f$,
  so that $L_f \to 1$ and $|L_J < \infty|$ as $F\mn \to 0$ (the latter condition takes into
  account that $J \sim f^2$ at small $F\mn$).
  
  Assuming static spherical symmetry, we have, as before, only radial electric and magnetic 
  fields, and the SET has again the properties $T^r_t =0$ and $T^t_t = T^r_r$, and the metric 
  can be written in the form \rf{ds1}. Specifically,
\bearr                     \label{SET-JS}
		T^t_t = T^r_r = 2 (L_f + f L_J) E^2 + \half L - 4 B^2 E^2 L_J,
\nnnv
		 T^\theta_\theta = T^\phi_\phi = - 2 (L_f + f L_J) B^2 + \half L - 4 B^2 E^2 L_J,
\ear
  with $E^2 = F^{tr}F_{rt}$ and $B^2 = F^{\theta\phi}F_{\theta\phi}$. The field equations 
  \rf{Max4} lead to
\beq       \label{QF4}
   		Q^{tr} = 4 (L_f +f L_J) F^{tr} = \frac{q_e}{r^2}, \qq  F_{\theta\phi} = q_m \sin\theta.
\eeq  

  Let us prove that the same no-go theorems as in $L(f)$ theories coupled to GR, are valid in the 
  theories \rf{S4}. 
% =========  
\begin{theorem}        \label{theo4} 
  The theories \rf{S4} do not admit \ssph\ electric solutions ($q_e\ne 0,\ q_m =0$) with a regular center 
  and a correct MWF limit. 
\end{theorem}    
\begin{proof}
  As with $L(f)$ theories in Section \ref{sec4}, assuming regularity at $r=0$, we must require that all 
  components of $T\mN$ should be finite, as well as their linear combinations. In particular, we require that
\beq                    \label{QQ1}
		|T^t_t - T^\theta_\theta| =  2 |L_f - 2E^2 L_J| E^2 < \infty.
\eeq 
  On the other hand, from \rf{QF4} we obtain 
\beq                    \label{QQ2}
		Q^{tr} Q_{rt} = 16 E^2 (L_f - 2E^2 L_J)^2 = q_e^2 /r^4 \to \infty\ \ \ {\rm as}\ r\to 0.
\eeq   
  The conditions \rf{QQ1} and \rf{QQ2} are only compatible if $E\to 0$ (that is, the field becomes weak)
  and $|L_f - 2E^2 L_J| \to \infty$, contrary to the desirable MWF limit. This completes the proof.
\qed
\end{proof}  
% ==========  
\begin{theorem}          \label{theo5}
  The theories \rf{S4} do not admit \ssph\ dyonic solutions ($q_e\ne 0,\ q_m \ne 0$) with a regular center. 
\end{theorem}    
\begin{proof}
    The same requirement as in the previous theorem, $|T^t_t - T^\theta_\theta| < \infty$, 
    necessary to be valid at a regular center, now reads
\beq                    \label{QQ3}
		|T^t_t - T^\theta_\theta| =  2 |L_f +f L_J| (E^2 + B^2) < \infty.
\eeq 
  Moreover, since both $E^2 > 0$ and $B^2 >0$, this inequality should hold with each of them taken
  separately. Applying it with $E^2$ together with the first equality \rf{QF4} that now leads to  
\beq                    \label{QQ4}
		Q^{tr} Q_{rt} = 16 (L_f +f L_J)^2 = q_e^2 /r^4 \to \infty\ \ \ {\rm as}\ r\to 0,
\eeq   
  we obtain, as before, $E\to 0$ and $L_f +f L_J \to \infty$. The same condition \rf{QQ3} with $B^2$
  leads to $L_f +f L_J \to 0$ (since $B^2 = q_m^2/r^4 \to \infty$). The resulting contradiction proves 
  the theorem..
\qed
\end{proof}    
% ==========
    
  It is also evident than none of the CNED theories can produce a regular BH or, more generally, 
  a solution with a regular center. Indeed, since the SET is proportional to that of Maxwell electrodynamics,
  it cannot lead to any \ssph\ metric other than \RN, though certainly the interpretation of its constants
  $m$ and $q$ will be different.
  
  As to purely magnetic solutions ($q_e\ne 0,\ q_m \ne 0$), in which $f = 2q_m^2/r^2$ and 
  $J = 2 q_m^4/r^8$, a regular center is possible under the condition that $L(f, J)$ tends to a 
  finite constant as both $f$ and $J$ tend to infinity. The whole situation looks quite the same 
  as with $L(f)$ theories.
  
  Let us give a confirming example, taking as a basis \eqn{L-kru20} for $L(f)$:
\beq                 \label{exQQ}
		L(f,J) = \frac{f}{1+ af/2} + \frac{bJ}{1+c J/2}, \qq   a,b,c = \const > 0,
\eeq  
  This Lagrangian has a correct MWF limit and tends to a finite limit at large $f$ and $J$.
  Since with $q_e=0$, according to \rf{SET-JS}, the density is simply $\rho = L/2$, the metric function 
  $A(r)$ is found as
\bearr                                \label{A-exQQ}
      A(r) = 1 - \frac{2 M(r)}{r}, \qq M(r) = M_1(r)+ M_2(r),
\nnn      
      M_1(r) = \frac{q^2}{2} \int \frac {r^2 dr}{r^4 + aq^2}, \qq
	 M_2(r) = \frac{bq^4}{2} \int \frac {r^2 dr}{r^8 + cq^4}, 
\ear
  where $q = q_m$. Integration gives 
\bearr                 \label{M-exQQ}
          M_1(r) = \frac{q^2} {8 h\sqrt 2}
          \bigg[2 \arctan \frac{h + \sqrt{2} r}{h} - 2 \arctan \frac{h - \sqrt{2} r}{h}
                 + \log \frac {h^2 - \sqrt{2} h r + r^2}{h^2 + \sqrt{2} h r + r^2}\bigg],
\nnn
          M_2(r) = \frac {bq^4}{16 j^5}          
          \bigg[2 C \Big(\arctan \frac{r +j C}{j S} + \arctan \frac{r - j C}{j S}\Big) 
\nnn          \inch
          		 -2 S \Big(\arctan \frac{r +j S}{j C} + \arctan \frac{r - j S}{j C}\Big)	
\nnn   \inch       		 
			  + S \,\log \frac{j^2 + 2 C j r + r^2}{j^2 - 2 C j r + r^2}
			  + C \, \log \frac{j^2 - 2 S j r + r^2}{j^2 + 2 S j r + r^2}\bigg],
\ear  
  where we have denoted $h = (aq^2)^{1/4},\ j = (cq^4)^{1/8},\ S=\sin(\pi/8),\ C= \cos(\pi/8)$. 
  At $r\to \infty$ we obtain
\beq                 \label{as-exQQ}
               M_1(r) \to \frac{\pi q^{3/2}}{4\sqrt{2}a^{1/4}},   \quad\ 
               M_2(r) \to \frac{\pi b q^{3/2} (C-S)}{8 c^{5/8}} , \quad\
               M = \lim_{r\to \infty}(M_1+M_2),
\eeq  
  where $M$ is the \Scw\ mass of completely \emag\ origin. At the center, we have
\beq                 \label{c-exQQ}
		    M_1(r) \approx \frac{r^3}{6a},   \quad\ M_2(r) \approx \frac{br^3}{6c}\quad\
		    {\rm as}\ \ r\to 0,
\eeq  	
   which leads to $A(r) = 1 + \cO(r^2)$, satisfying the regular center condition \rf{reg-c}. 
   We have obtained a regular \asflat\ solution in $L(f,J)$ NED with a correct WMF limit. 
   
   Is it a \bh\ solution? To make it clear, let us fix the parameters $q=1,\ a=1,\ c=1$, then the only 
   remaining free parameter is $b$, and
\beq
			M = \frac{\pi}{8}[\sqrt{2} + b(C-S)].
\eeq   
  The behavior of $A(r)$ at different values of $b$ is shown in Fig.\,\ref{fig4}.
% --------------------------- fig 4
\begin{figure}
\centering
\includegraphics[width=12cm]{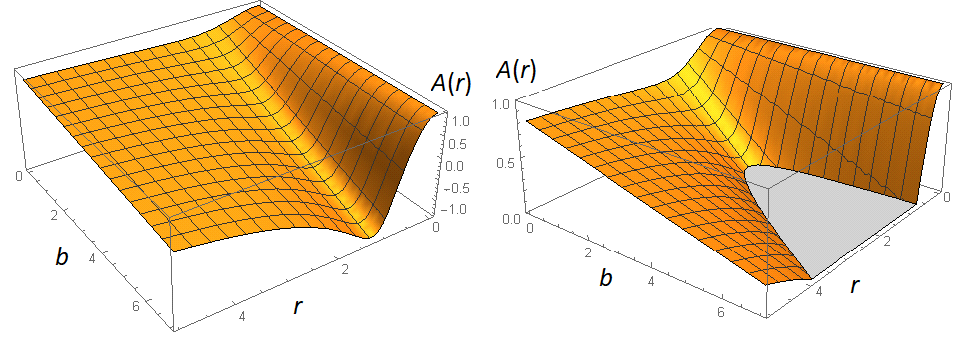} 
\caption{\small
		The behavior of $A(r)$ in the magnetic solution for the theory \rf{exQQ} with $q=a=c=1$
		and different values of $b$. The left panel shows the full range of $A(r)$, while in the right 
		one the same dependence is shown with a gray ``floor'' on the level $A=0$, visualizing 
		the places where $A$ becomes zero or negative, and the solution then describes a \bh.}
		\label{fig4}
\end{figure}  
% ----------------------------         
   	An inspection shows that at small $b$ the solution is of solitonic nature, at $b \approx 2.436$
   	(corresponding to $M \approx 1.073$) there emerges a single extremal horizon, and at larger 
   	$b$ (large $M$) we obtain a regular \bh\ with two simple horizons.  
      	
% ============================
\section{Conclusion}\label{sec6}
% ============================

  We have discussed the opportunities of obtaining regular \sph\ \bh\ solutions in GR sourced 
  by nonlinear \emag\ fields governed by different NED theories. It happens that such NED 
  black holes (as well as solitons) with a regular center in $L(f)$ theories can exist with pure 
  electric or pure magnetic charges, and only systems with a magnetic charge are compatible 
  with Lagrangians having a correct Maxwell behavior at small $f$. Dyonic configurations, 
  with $q_e\ne 0$ and $q_m\ne 0$, cannot contain a regular center, whatever be the 
  function $L(f)$. 
  
  In Maxwell's electrodynamics there is the well-known symmetry (duality) between 
  electric and magnetic fields, leading to the same symmetry between the corresponding 
  solutions to the Einstein-Maxwell equations, at least in the absence of currents ans charges.
  Unlike that, in NED, we only have FP duality that connects purely electric and purely
  magnetic configurations with the same metric but sourced by different NED theories.
  Accordingly, in a theory specified by a particular function $L(f)$, the properties of electric
  and magnetic solutions are quite different.

  It turns out that magnetic solutions lead to completely regular configurations, while for their 
  electric counterparts, obtained from them using FP duality and well-behaved in the framework 
  of the ``Hamiltonian'' formulation of NED, the Lagrangian formulation is ill-defined, and the 
  behavior of NED photons exhibits undesired features at some intermediate radii: they
  experience an infinite blueshift, indicating an instability of such a background configuration.
  
  The dynamic stability of regular magnetic solutions is also questionable since one of the 
  sufficient stability conditions ($\Phi >0 $) is inevitably violated near a regular center.
  Thus general stability results for regular \bhs\ probably cannot be obtained, and stability 
  studies of individual solutions seem to be necessary.

  We here did not touch upon thermodynamic properties of NED \bhs, this important issue is
  discussed in many papers, see, among others, \cite{Bret-05, Kru-16, FW-16, kru-21, balart-21} 
  and references therein. Let us only remark here that a thermodynaic instability of \bhs\ 
  related to their negative heat capacity is implemented in the process of Hawking evaporation,
  which is very slow for sufficiently large \bhs\ and can be practically ignored for \bhs\ 
  with stellar and larger masses, irrespective of their global regularity properties.

  There are many results obtained with more general NED Lagrangians, such as $L(f,h)$
  and $L(f, J)$. Many no-go theorems concerning possible regular \bhs\ with $L(f, h)$ NED
  have been presented in Refs.\,\cite{bokulic-21, bokulic-22}. As to $L(f, J)$ NED, involving
  a fourth-order \emag\ invariant, we have verified here that the restrictions on regular \bh\
  existence obtained with $L(f)$ are extended to this class of theories without change. In 
  particular, regular \bhs\ with a magnetic charge can also be obtained, as we have 
  confirmed with an explicit example. Very probably these results can be further extended to 
  include \emag\ invariants of still higher orders constructed in the same manner as 
  $J = F\mn F^{\nu\rho} F_{\rho\sigma} F^{\sigma\mu}$. 
  
  Among the remaining theoretic problems deserving further studies let us mention the 
  dynamic stability problem for regular magnetic \bhs\ and the causality issue related to the 
  predicted superluminal velocities of NED photons.


\begin{thebibliography}{99}

\bibitem{Arel-06}
	A. V. B. Arellano and F. S. N. Lobo,
         Evolving wormhole geometries within nonlinear electrodynamics,
         Class. Quantum Grav. {\bf 23}, 5811 (2006); gr-qc/0608003.

\bibitem{Arel-09}
	A. V. B. Arellano, N. Bret\'on, and R Garcia-Salcedo,
         Some properties of evolving wormhole geometries within nonlinear electrodynamics,
         {\it Gen. Rel. Grav.} {\bf 41}, 2561 (2009);  arXiv: 0804.3944.
         
\bibitem{gad-98}
	E. Ayon-Beato and A. Garcia Diaz, 
	Regular black hole in general relativity coupled to nonlinear electrodynamics,
	Phys. Rev. Lett. {\bf 80}, 5056 (1998).      
	
\bibitem{balart-21}
	Leonardo Balart, Sharmanthie Fernando,
	Thermodynamics and heat engines of black holes with Born-Infeld-type electrodynamics,
	Mod. Phys. Lett. A {\bf 36}, 2150102 (2021); arXiv: 2103.15040; doi: 10.1142/S0217732321501029	   
         
\bibitem{vag-14}
	 L. Balart and E. C. Vagenas,
	 Regular black hole metrics and the weak energy condition
	 {\it Phys. Lett. B} {\bf 730}, 14 (2014); arXiv: 1401.2136.          

\bibitem{Bam-13}
	C. Bambi and L. Modesto,
	Rotating regular black holes,
         {\it Phys. Lett. B} {\bf 721}, 329 (2013);   arXiv: 1302.6075.
         
\bibitem{bandos-20}	
	I. Bandos, K. Lechner, D. Sorokin, and P. Townsend,
	(ModMax)
	Phys. Rev. D 102, 121703 (2020), arXiv: 2007.09092.

\bibitem{Mat-02}
	W. Berej and J. Matyjasek,
         Vacuum polarization in the spacetime of charged nonlinear black hole,
         {\it Phys. Rev. D} {\bf 66}, 024022 (2002); gr-qc/0204031.         

\bibitem{Boe-07}
	Ch. G. Boehmer, T. Harko, and F. S. N. Lobo,
         Conformally symmetric traversable wormholes,
         {\it Phys. Rev. D} {\bf 76},  084014 (2007);  arXiv: 0708.1537.

\bibitem{bokulic-21}	
	A. Bokuli\'c, I. Smoli\'c, and T. Juri\'c,
	Nonlinear electromagnetic fields in strictly stationary spacetimes,
	Phys. Rev. D 105, 024067 (2022); arXiv: 2111.10387
	
\bibitem{bokulic-22}
	A. Bokuli\'c, I. Smoli\'c, and T. Juri\'c,
	Constraints on singularity resolution by nonlinear electrodynamics,
	Phys. Rev. D {\bf 106}, 064020 (2022); arXiv: 2206.07064.
	
\bibitem{B-Inf}
     M. Born and L. Infeld, 
	Foundations of the new field theory, Nature {\bf 132}, 1004  (1933);
	{\it Proc. R. Soc. Lond.} {\bf 144}, 425 (1934).
	
\bibitem{Bret-05}
	N. Bret\'on,
      Smarr's formula for black holes with non-linear electrodynamics,
	{\it Gen. Rel. Grav.} {\bf 37}, 643 (2005);  gr-qc/0405116.
	
\bibitem{Bret-05s}
	N. Bret\'on,  Stability of nonlinear magnetic black holes,
      {\it Phys. Rev. D} {\bf 72}, 044015 (2005); hep-th/0502217.

\bibitem{breton-14}
	Nora Bret\'on, Santiago Esteban Perez Bergliaffa,
	On the stability of black holes with nonlinear electromagnetic fields, 	arXiv: 1402.2922.

\bibitem{B-comment} 
	K. A. Bronnikov,
	Comment on `Regular black hole in general relativity coupled to nonlinear electrodynamics' , 
	{\it Phys. Rev. Lett.} {\bf 85}, 4641 (2000). 

\bibitem{k-NED}
	K. A. Bronnikov,
	Regular magnetic black holes and monopoles from nonlinear electrodynamics,
	\PRD {63} 044005 (2001); gr-qc/0006014.

\bibitem{17-dyon}
	K. A. Bronnikov, Dyonic configurations in nonlinear electrodynamics coupled to general 
	   relativity, \GC {23} 343 (2017); arXiv: 1708.08125.
	   
\bibitem{k-18}
	K. A. Bronnikov, Nonlinear electrodynamics, regular black holes and wormholes,
	Int. J. Mod. Phys. D {\bf 27}, 1841005  (2018);  arXiv: 1711.00087.	   

\bibitem{k-22b}
	K. A. Bronnikov, Black bounces, wormholes, and partly phantom scalar fields,
	Phys. Rev. D {\bf 106}, 064029 (2022); arXiv: 2206.09227.
	
\bibitem{B-Shi2}         
	K. A. Bronnikov, V. N. Melnikov, G. N. Shikin, and K. P. Staniukovich, 
	Scalar, electromagnetic, and gravitational fields interaction: particlelike solutions,
	{\it Ann. Phys. (N.Y.)} {\bf 118}, 84 (1979).
	
\bibitem{BR-book}
     K A. Bronnikov and S. G. Rubin,
           {\it Black Holes, Cosmology, and Extra Dimensions} (World Scientific, 2012).

\bibitem{B-Shi}
	K. A. Bronnikov and G. N. Shikin, 
	On the Reissner-Nordstr\"om problem with a nonlinear electromagnetic field,
         in {\it Classical and Quantum Theory of Gravity} (Trudy IF AN BSSR, Minsk, 1976), 
         p. 88 (in Russian).
        
\bibitem{we-02}
	K. A. Bronnikov, G. N. Shikin, and E. N. Sibileva,
	Self-gravitating stringlike configurations from nonlinear electodynamics,
         {\it Grav. Cosmol.} {\bf 9},  169 (2003);  gr-qc/0308002.   
         
\bibitem{k-22a}
	K. A. Bronnikov and R. K. Walia, Field sources for Simpson-Visser space-times,
	Phys. Rev. D {\bf 105}, 044039 (2022); arXiv: 2112.13198.
  
\bibitem{Bur1}
	A. Burinskii and S. R. Hildebrandt, 
	New type of regular black holes and particlelike solutions from nonlinear electrodynamics,
	{\it Phys. Rev. D} {\bf 65}, 104017 (2002); hep-th/0202066.

\bibitem{canate-22}
	Pedro Ca\~nate,
	Black-bounces as magnetically charged phantom regular black holes in Einstein-nonlinear 
	electrodynamics gravity coupled to a self-interacting scalar field,
	Phys. Rev. D {\bf 106}, 024031 (2022); arXiv: 2202.02303.	
	
\bibitem{denis-19} 
	I.P. Denisova, B.D. Garmaev, and V.A. Sokolov,
      Compact objects in conformal nonlinear electrodynamics,
	 Eur. Phys. J. C {\bf 79}, 531 (2019); arXiv: 1901.05318.	
  	
\bibitem{Dym-15a}
	I. Dymnikova and E. Galaktionov,
         Regular rotating electrically charged black holes and solitons in nonlinear electrodynamics 
         minimally coupled to gravity,
         Class. Quantum Grav. {\bf 32}, 165015 (2015); arXiv: 1510.01353.

\bibitem{FW-16} 
	Zhong-Ying Fan and Xiaobao Wang,
	Construction of regular black holes in general relativity,
	 \PRD {94} 124027 (2016); arXiv: 1610.02636.

\bibitem{tse-1}
	E.S. Fradkin and A.A. Tseytlin, Nonlinear electrodynamics from quantized strings, 
	Phys. Lett. B {\bf 163}, 123 (1985).
	
\bibitem{fran-21}
	E.~Franzin, S.~Liberati, J.~Mazza, A.~Simpson and M.~Visser,
	Charged black-bounce spacetimes, JCAP \textbf{07}, 036 (2021).
	
\bibitem{gao-21}
	Changjun Gao,
	Black holes with many horizons in the theories of nonlinear electrodynamics,
	Phys. Rev. D {\bf 104}, 064038 (2021); arXiv: 2106.13486; doi: 10.1103/PhysRevD.104.064038
	
\bibitem{gdiaz-22}
	Alberto A. Garcia-Diaz,	
	AdS-dS stationary rotating black hole exact solution within Einstein-nonlinear electrodynamics, 
	Annals Phys. {\bf 441}, 168880 (2022); arXiv: 2201.10682.		
			
\bibitem{hadamard}
	J. Hadamard, in Le\c{c}ons sur la propagation des ondes et
	les \'equations de l'hydrodynamique, (Ed. Hermann, Paris, 1903).
	
\bibitem{EuH}
     W. Heisenberg and H. Euler, Folgerungen aus der Diracschen Theorie des Positrons, 
     {\it Z. Phys.} {\bf 98}, 714 (1936).
     
\bibitem{sangpyo-22}
	Chul Min Kim and Sang Pyo Kim,
	Vacuum birefringence in a supercritical magnetic field and a subcritical electric field,
	arXiv: 2202.05477.     
     
\bibitem{krug-15}	
	S. I. Kruglov, Nonlinear electrodynamics with birefringence,
	Phys. Lett. A 379, 623 (2015); arXiv: 1504.03535;
	doi: 10.1016/j.physleta.2014.12.026

\bibitem{kru-16a}	
	S. I. Kruglov, Acceleration of Universe by nonlinear electromagnetic fields,
	arXiv: 1603.07326, doi: 10.1142/S0218271816400022

\bibitem{kru-16b}	
	S. I. Kruglov,
	Nonlinear arcsin-electrodynamics and asymptotic \RN\ black holes,
	Ann. Physik.(Berlin), {\bf 528}, 588 (2016); arXiv: 1607.07726;
	doi: 10.1002/andp.201600027

\bibitem{Kru-16}
	S. I. Kruglov, Asymptotic Reissner-Nordstr\"om solution within nonlinear electrodynamics,
         {\it Phys. Rev. D} {\bf 94}, 044026 (2016); arXiv: 1608.04275.   

\bibitem{kru-19}
	S. I. Kruglov, Dyonic black holes with nonlinear logarithmic electrodynamics,
	Grav. Cosmol. {\bf 25}, 190--195 (2019); arXiv: 1909.05674; doi: 10.1134/S0202289319020105
	
\bibitem{krugl-19a}	
	S. I. Kruglov, Dyonic black holes in framework of Born--Infeld-type electrodynamics,
	Gen. Rel. Grav. {\bf 51}, 121 (2019); arXiv: 1909.11661.

\bibitem{krugl-19b}	
	S. I. Kruglov, Dyonic and magnetic black holes with nonlinear arcsin-electrodynamics,
	Annals Phys. {\bf 409 }, 167937 (2019); arXiv: 1911.04253.	      

\bibitem{kru-20}
	S. I. Kruglov, Dyonic and magnetized black holes based on nonlinear electrodynamics,
	Eur. Phys. J. C  {\bf 80}, 250 (2020); arXiv: 2003.10845; doi: 10.1140/epjc/s10052-020-7809-x
	
\bibitem{kru-21}	
	S. I. Kruglov, Remarks on nonsingular models of Hayward and magnetized black hole
	with rational nonlinear electrodynamics, 
	Grav. Cosmol. {\bf 27}, 78--84 (2021); arXiv: 2103.14087; doi: 10.1134/S0202289321010126

\bibitem{kubi-22}
	David Kubiznak, Tayebeh Tahamtan, and Otakar Svitek,
	Slowly rotating black holes in nonlinear electrodynamics,
	Phys. Rev. D {\bf 105}, 104064 (2022); arXiv: 2203.01919. 
	
\bibitem{LL}
	L. Landau and E. Lifshitz, {\it Classical Theory of Fields } (3rd ed., Pergamon, London, 1971).
        
\bibitem{Jin-14}
	Jin Li, Kai Lin and Nan Yang,
        Nonlinear electromagnetic quasinormal modes and Hawking radiation of a regular 
         black hole with magnetic charge,
        {\it Eur. Phys. J. C} {\bf 75}, 131 (2015); arXiv: 1409.5988.           
        
\bibitem{lobo-20}
	F.S.N.~Lobo, M.E.~Rodrigues, M.V.d.S.~Silva, A.~Simpson, and M.~Visser,
	Novel black-bounce spacetimes: wormholes, regularity, energy conditions, and causal structure,
	Phys. Rev. D \textbf{103}, 084052 (2021).
	       
\bibitem{Mat-13}
	J. Matyjasek, P. Sadurski, and D. Tryniecki,
         Inside the degenerate horizons of regular black holes,
	{\it Phys. Rev. D} {\bf 87}, 124025 (2013); arXiv: 1304.6347.
    
\bibitem{Mat-09}
	J. Matyjasek, D. Tryniecki, and M. Klimek,
         Regular black holes in an asymptotically de Sitter universe,
         {\it Mod. Phys. Lett. A} {\bf 23}, 3377 (2009);  arXiv: 0809.2275.        

\bibitem{tse-2}	
	R.R. Metsaev, M. Rakhmanov and A.A. Tseytlin, The Born-Infeld action as the 
	effective action in the open superstring theory, 
	Phys. Lett. B {\bf 193}, 207 (1987).
	
\bibitem{mkr-22}	
	Karapet Mkrtchyan and Mantas Svazas,
	Solutions in nonlinear electrodynamics and their double copy regular black holes,
	JHEP (9) 012 (2022); arXiv: 2205.14187; doi: 10.1007/JHEP09(2022)012
	
\bibitem{Mor-03}
	C. Moreno and O. Sarbach,
         Stability properties of black holes in self-gravitating nonlinear electrodynamics,
	{\it Phys. Rev. D} {\bf 67},  024028 (2003); gr-qc/0208090; doi: 10.1103/PhysRevD.67.024028

\bibitem{od-17}	
	Shin'ichi Nojiri and S.D. Odintsov, Regular multi-horizon black holes in modified gravity with 
	non-linear electrodynamics, Phys. Rev. D {\bf 96}, 104008 (2017); arXiv: 1708.05226; 
	doi: 10.1103/PhysRevD.96.104008
	
\bibitem{soda-20}
	Kimihiro Nomura, Daisuke Yoshida, Jiro Soda,
	Stability of magnetic black holes in general nonlinear electrodynamics,
	Phys. Rev. D 101, 124026 (2020); arXiv: 2004.07560; doi: 10.1103/PhysRevD.101.124026

\bibitem{nov-1}
     M. Novello, V. A. de Lorenci, J. M. Salim, and R. Klippert,
	Geometrical aspects of light propagation in nonlinear electrodynamics,
         {\it Phys. Rev. D} {\bf 61}, 045001 (2000).

\bibitem{nov-2}
     M. Novello, S. E. Perez Bergliaffa, and J. M. Salim, 
     Singularities in General Relativity coupled to nonlinear electrodynamics,
     {\it Class. Quantum Grav.} {\bf 17}, 3821 (2000);   gr-qc/0003052.

\bibitem{Pel-T}
	R. Pellicer and R. J. Torrence,
	Nonlinear electrodynamics and general relativity,
	{\it J. Math. Phys.} {\bf 10}, 17+18 (1969).

\bibitem{Pleb}
	J. Plebanski, {\it Non-Linear Electrodynamics --- A Study} 
	(C.I.E.A. del I.P.N., Mexico City, 1966).
	
\bibitem{rosen-52}	
     N. Rosen and H.B. Rosenstock, The force between particles in a nonlinear field theory,
     Phys. Rev. {\bf 85}, 257 (1952); doi: 10.1103/PhysRev.85.257
	
\bibitem{ryb-book}	
	Yu.P. Rybakov, {\it Particle Structure in Nonlinear Field Theory}
	(Peoples' Friendship University Press, Moscow, 1985).
	
\bibitem{sala-87}
        I. H. Salazar, A. Garcia and J. Plebanski, 
        Duality rotations and type D solutions to Einstein equations with 
        nonlinear electromagnetic sources.
        {\it J. Math. Phys.} {\bf 28}, 2171 (1987).

\bibitem{seiberg}
	N. Seiberg and E. Witten, String theory and noncommutative geometry,
     {\it J. High Energy Phys.} {\bf 09}, 032 (1999); hep-th/9908142.
     
\bibitem{usov-11}
	Anatoly E. Shabad, Vladimir V. Usov, 
	Effective Lagrangian in nonlinear electrodynamics and its properties of causality and unitarity,
	Phys. Rev. D {\bf 83}, 105006 (2011); doi: 10.1103/PhysRevD.83.105006

\bibitem{simp-18}
	A.~Simpson and M.~Visser, Black bounce to traversable wormhole, 
	JCAP \textbf{02}, 042  (2019).    
     
\bibitem{sokolov-21}	
	V.A. Sokolov, Extended duality condition for conformal vacuum nonlinear electrodynamics,
	Phys. Rev. D {\bf 104}, 124035 (2021); doi: 10.1103/PhysRevD.104.124035	
	
\bibitem{sorokin-21}
	Dmitri P. Sorokin,		
	Introductory notes on nonlinear electrodynamics and its applications,	
	Fortschritte der Physik {\bf 70} (7--8), 2200092 (2022); arXiv: 2112.12118.
	
\bibitem{tosh-17}
	Bobir Toshmatov, Zdenek Stuchlik, and Bobomurat Ahmedov,
	Generic rotating regular black holes in general relativity coupled to nonlinear electrodynamics,
	Phys. Rev. D {\bf 95}, 084037 (2017); arXiv: 1704.07300. 
	
\bibitem{tosh-19}
	Bobir Toshmatov, Zden\v{e}k Stuchl\'{\i}k, Bobomurat Ahmedov, Daniele Malafarina,
	Relaxations of perturbations of spacetimes in general relativity coupled to nonlinear electrodynamics,
	Phys. Rev. D 99, 064043 (2019); arXiv: 1903.03778; doi: 10.1103/PhysRevD.99.064043
	
\bibitem{vil-sh}
	A. Vilenkin and E.P.S. Shellard,
	{\it Cosmic Strings and Other Topological Defects} (Cambridge Univ. Press, Cambridge, 1994).
	     
\bibitem{yang-22}	
	Yisong Yang, Dyonically charged black holes arising in generalized Born--Infeld theory of electromagnetism,
	Annals Phys. {\bf 443}, 168996 (2022); arXiv: 2204.11313; doi: 10.1016/j.aop.2022.168996
	    
\end{thebibliography}
\end{document}